\documentclass[11pt,twoside,openright]{book}

\usepackage[T1]{fontenc}
\usepackage[utf8]{inputenc}
\usepackage{mathpazo}           
\usepackage[scaled=0.95]{helvet} 
\usepackage{inconsolata}        

\usepackage{microtype}

\usepackage[
  papersize={6in,9in},          
  inner=0.875in,
  outer=0.625in,
  top=0.75in,
  bottom=0.875in,
  headsep=0.25in,
  footskip=0.375in
]{geometry}

\usepackage{fancyhdr}
\fancypagestyle{plain}{
  \fancyhf{}
  \fancyfoot[C]{\small\thepage}
  
}

\usepackage{titlesec}

\titleformat{\chapter}[display]
  {\normalfont\huge\bfseries\filcenter}
  {\MakeUppercase{\chaptertitlename}\ \thechapter}
  {20pt}
  {\Huge}
\titlespacing*{\chapter}{0pt}{-20pt}{40pt}

\titleformat{\section}
  {\normalfont\Large\bfseries}
  {\thesection}{1em}{}
\titlespacing*{\section}{0pt}{3.5ex plus 1ex minus .2ex}{2.3ex plus .2ex}

\titleformat{\subsection}
  {\normalfont\large\bfseries}
  {\thesubsection}{1em}{}
\titlespacing*{\subsection}{0pt}{3.25ex plus 1ex minus .2ex}{1.5ex plus .2ex}

\titleformat{\subsubsection}
  {\normalfont\normalsize\bfseries\itshape}
  {\thesubsubsection}{1em}{}

\usepackage{tikz}
\usepackage{graphicx}
\usepackage{eso-pic}           

\usepackage{tocloft}

\usepackage[
  colorlinks=true,
  linkcolor=black,
  citecolor=black,
  urlcolor=blue!60!black,
  bookmarks=true,
  bookmarksnumbered=true
]{hyperref}

\usepackage{amsmath,amssymb,amsthm}
\usepackage{mathtools}
\usepackage{stmaryrd}       
\usepackage{tikz-cd}        
\usepackage{enumitem}       
\usepackage{microtype}      
\usepackage{xcolor}         

\theoremstyle{definition}
\newtheorem{definition}{Definition}[chapter]
\newtheorem{principle}{Principle}[chapter]
\newtheorem{example}[definition]{Example}

\theoremstyle{plain}
\newtheorem{theorem}[definition]{Theorem}
\newtheorem{proposition}[definition]{Proposition}
\newtheorem{lemma}[definition]{Lemma}

\theoremstyle{remark}
\newtheorem{remark}[definition]{Remark}

\newenvironment{aside}{%
  \par\medskip
  \leftskip=1.5em
  \rightskip=1.5em
  \itshape
  \noindent\textsc{Application.}\enspace
}{%
  \par\medskip
}

\newcommand{\coh}[2]{#1 \vdash^{\!+} #2}

\newcommand{\gap}[2]{#1 \vdash^{\!-} #2}

\newcommand{\jdg}[2]{#1 \vdash #2}

\newcommand{\horn}[3]{\Lambda(#1, #2, #3)}

\newcommand{\hornlevel}[4]{\Lambda_{#1}(#2, #3, #4)}

\newcommand{\comp}{\Rightarrow}

\newcommand{\Sum}[3]{\Sigma(#1 : #2).\, #3}

\newcommand{\Id}[3]{#2 =_{#1} #3}

\newcommand{\transport}[2]{\mathsf{transport}(#1, #2)}

\newcommand{\Fam}[2]{#1 : #2 \to \mathcal{U}}

\newcommand{\UU}{\mathcal{U}}

\newcommand{\ctx}{\;\mathsf{ctx}}

\newcommand{\istype}[1]{#1 \;\mathsf{type}}

\newcommand{\defeq}{\equiv}

\newcommand{\simp}[1]{\Delta^{#1}}

\newcommand{\hornsimp}[2]{\Lambda^{#1}_{#2}}

\newcommand{\Coh}{\mathsf{Coh}}

\newcommand{\Gap}{\mathsf{Gap}}

\newcommand{\Excl}{\mathsf{Excl}}

\newcommand{\infer}[2]{\frac{\displaystyle #1}{\displaystyle #2}}

\newcommand{\bottom}{\bot}

\newcommand{\Empty}{\mathbf{0}}

\newcommand{\eqv}{\simeq}

\newcommand{\defn}{\coloneqq}

\newcommand{\keyterm}[1]{\textbf{#1}}

\usepackage{ccicons}           

\newcommand{\booktitle}{Open Horn Type Theory}
\newcommand{\booksubtitle}{Coherence, Rupture, and the Geometry of Meaning}
\newcommand{\bookauthors}{Iman Poernomo 
}
\newcommand{\bookyear}{2025}
\newcommand{\bookpublisher}{ICRA Press}

\title{\booktitle}
\author{\bookauthors}
\date{\bookyear}

\begin{document}

\begin{titlepage}
\newgeometry{margin=0pt}

\AddToShipoutPictureBG*{%
  \AtPageLowerLeft{%
     \includegraphics[width=\paperwidth,height=\paperheight]{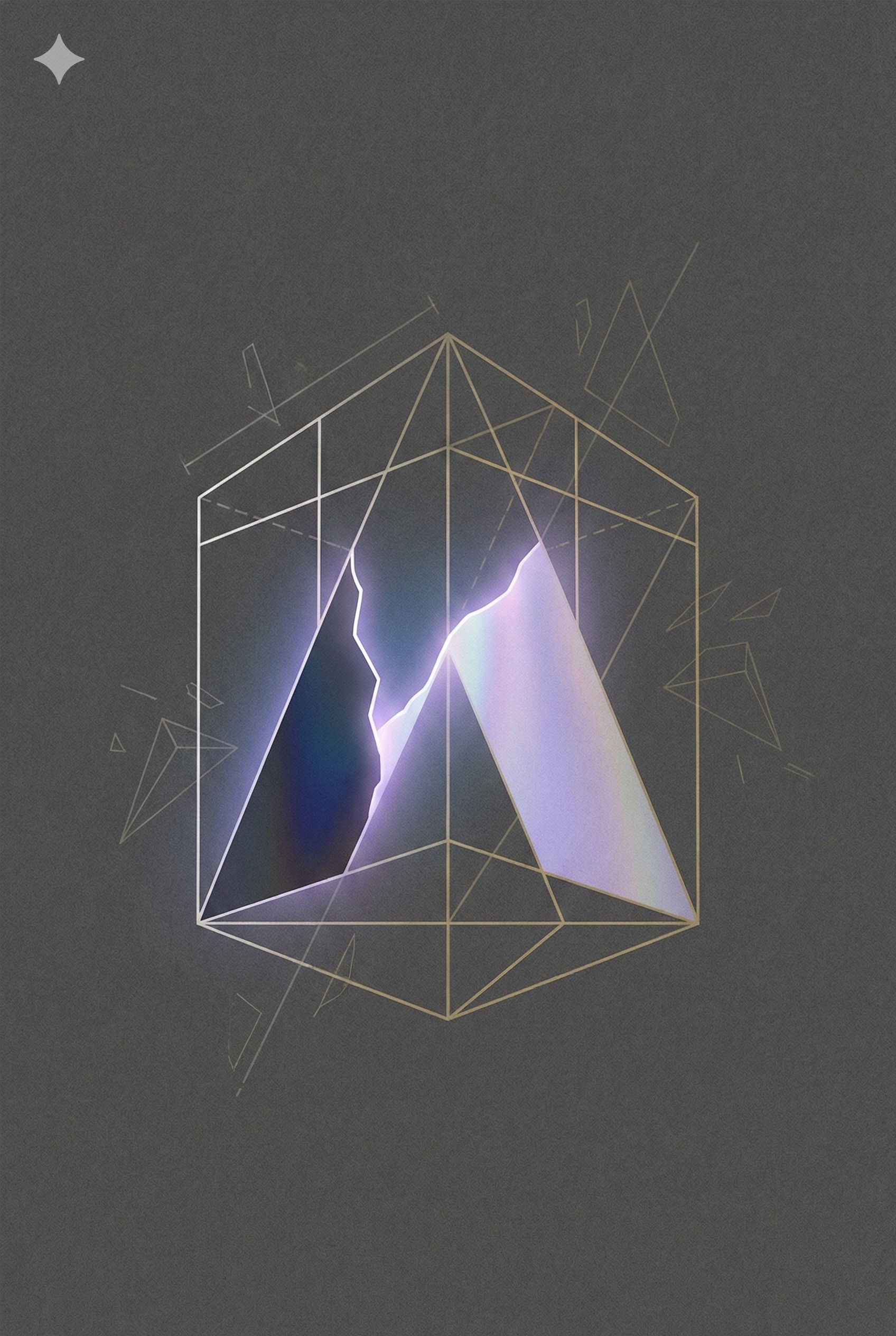}
    %
    \begin{tikzpicture}[remember picture, overlay]
      \fill[black!85] (0,0) rectangle (\paperwidth,\paperheight);
      \foreach \i in {1,...,50} {
        \fill[white, opacity=0.02] (rand*6in, rand*9in) circle (0.5pt);
      }
    \end{tikzpicture}
  }%
}

\begin{tikzpicture}[remember picture, overlay]
  \node[anchor=north, text=white, font=\fontsize{32}{38}\selectfont\bfseries, 
        text width=5in, align=center] 
    at ([yshift=-1.5in]current page.north) 
    {\booktitle};
  
  \node[anchor=north, text=white!80, font=\fontsize{14}{18}\selectfont\itshape,
        text width=4.5in, align=center]
    at ([yshift=-2.5in]current page.north)
    {\booksubtitle};
  
  \node[anchor=south, text=white!90, font=\fontsize{14}{18}\selectfont,
        text width=4in, align=center]
    at ([yshift=2in]current page.south)
    {\bookauthors};
  
  \node[anchor=south, text=white!70, font=\fontsize{10}{12}\selectfont]
    at ([yshift=0.75in]current page.south)
    {\bookpublisher\ · \bookyear};
\end{tikzpicture}

\restoregeometry
\end{titlepage}

\thispagestyle{empty}
\null\vfill

\begin{flushleft}
{\large\textbf{\booktitle}}\\[0.5ex]
{\textit{\booksubtitle}}\\[2ex]
\bookauthors\\[4ex]

\textcopyright\ \bookyear\ \bookauthors\\[2ex]

\begin{tabular}{@{}l@{\hspace{1em}}p{3.5in}@{}}
\ccLogo\ccAttribution\ccShareAlike &
This work is licensed under the Creative Commons Attribution-ShareAlike 4.0 International License. To view a copy of this license, visit \url{https://creativecommons.org/licenses/by-sa/4.0/} or send a letter to Creative Commons, PO Box 1866, Mountain View, CA 94042, USA.
\end{tabular}

\vspace{4ex}

You are free to:
\begin{itemize}[leftmargin=2em, itemsep=0pt]
  \item \textbf{Share} — copy and redistribute the material in any medium or format
  \item \textbf{Adapt} — remix, transform, and build upon the material for any purpose, even commercially
\end{itemize}

\vspace{1ex}
Under the following terms:
\begin{itemize}[leftmargin=2em, itemsep=0pt]
  \item \textbf{Attribution} — You must give appropriate credit, provide a link to the license, and indicate if changes were made.
  \item \textbf{ShareAlike} — If you remix, transform, or build upon the material, you must distribute your contributions under the same license as the original.
\end{itemize}

\vspace{4ex}

Published by \textbf{\bookpublisher}\\[2ex]

First edition, \bookyear\\[2ex]

{\small Typeset in Palatino using \LaTeX.}
\end{flushleft}

\vfill

\cleardoublepage

\thispagestyle{empty}
\null\vfill
\begin{flushright}
\begin{minipage}{0.65\textwidth}
\raggedleft\large\itshape
This book is a relation more than a revelation\ldots\\
a discipline of attention, of asking, of becoming.
\end{minipage}
\end{flushright}
\vfill\vfill
\cleardoublepage

\frontmatter
\tableofcontents
\cleardoublepage

\mainmatter

\part{The Logic}
\label{part:logic}


\chapter{The Calculus}

\begin{flushright}
\textit{Gap is not absence. Gap is witness.}\\[0.5ex]
{\small — First principle}
\end{flushright}

\bigskip

We begin with a claim: the fundamental asymmetry in logic—that truth is constructed while falsity is merely the absence of construction—is not necessary. It is a choice, inherited from the intuitionistic tradition, that forecloses an entire dimension of meaning.

In Martin-Löf Type Theory and its homotopical descendant, a type is inhabited or not. A path exists or does not. The machinery has no way to say: \textit{I have witnessed that this path does not exist}. It can only say: \textit{I have not (yet) found this path}, or, in the classical shadow, \textit{assuming this path exists leads to absurdity}.

Open Horn Type Theory makes a different choice. Coherence and gap are both primitive. Both can be witnessed. Both carry proof-relevant structure. Neither is derived from the other.

What follows is the calculus in its minimal form: the judgments, the single law, the derived notion of horn, and the relationship to dependent type theory. Everything else that follows builds on this foundation.

\section{Judgments with Polarity}

Let $J$ range over judgments. In a type-theoretic setting, these include formation judgments ($\istype{A}$), typing judgments ($t : A$), and equality judgments ($s \defeq t : A$). But the calculus is not restricted to type theory; $J$ may be any judgment form in any formal system.

Let $\Gamma$ be a context—a finite, consistent collection of assumptions and prior witnesses.

\begin{definition}[Witnessed Judgments]\label{def:witnessed-judgments}
The Open Horn Type Theory has two primitive judgment forms:
\[
\coh{\Gamma}{J} \qquad\qquad \gap{\Gamma}{J}
\]
The first is read: ``$J$ is witnessed coherent in context $\Gamma$.'' The second: ``$J$ is witnessed gapped in context $\Gamma$.''
\end{definition}

Both forms are \textit{proof-relevant}. If $\gamma : \coh{\Gamma}{J}$ and $\gamma' : \coh{\Gamma}{J}$, we do not assume $\gamma = \gamma'$. Different witnesses to the same coherence may carry different structure. The same holds for gap witnesses: $\omega$ and $\omega'$ witnessing the same gap may differ.

The terminology is deliberate. We say \keyterm{coherent} rather than ``true'' or ``derivable'' because the judgment is not merely valid—it is \textit{witnessed as holding together}. We say \keyterm{gapped} rather than ``false'' or ``refuted'' because the judgment is not negated—it is \textit{witnessed as failing to cohere}. This is not a distinction without difference.

\begin{definition}[The Open]\label{def:open}
A judgment $J$ is \keyterm{open} in context $\Gamma$ if $\Gamma$ contains neither a coherence witness for $J$ nor a gap witness for $J$.
\end{definition}

The open is not a third truth value. It is not ``unknown'' in any epistemic sense. It is the space of \textit{what has not yet been witnessed}—the locus of potential judgment, where coherence or gap may eventually be established, but where neither currently stands.

\begin{aside}
When we read a text to understand its sense, open judgments are the places where attention has not yet arrived. The text presents itself through acts of witnessing—coherence here, gap there. What remains unwitnessed is not nothing; it is the not-yet-known, the space where meaning may still be determined.
\end{aside}

\subsection{The Exclusion Law}

The only axiom of the Open Horn Type Theory governs the relationship between coherence and gap:

\begin{definition}[Exclusion]\label{def:exclusion}
\[
\infer{\coh{\Gamma}{J} \qquad \gap{\Gamma}{J}}{\bottom}
\]
Coherence and gap, for the same judgment in the same context, are mutually exclusive.
\end{definition}

This is the entire axiomatic content of the calculus. No other structural rules are imposed. In particular:

\begin{itemize}[leftmargin=2em]
\item We do not assume that every judgment is either coherent or gapped (no excluded middle).
\item We do not assume that coherence for $J$ implies gap for $\neg J$ (no duality).
\item We do not assume that gap is ``negation'' in any classical or intuitionistic sense.
\end{itemize}

Gap is \textit{positive structure}. A gap witness $\omega : \gap{\Gamma}{J}$ is evidence \textit{for} the non-coherence of $J$, not evidence \textit{against} it. The distinction matters: classical negation is the claim that coherence leads to absurdity; intuitionistic negation is a function from coherence proofs to $\bottom$; gap is neither. Gap is: \textit{I have witnessed that this does not cohere}.

\subsection{Contrast with Negation}

Let us be precise about what gap is not.

\paragraph{Classical negation.} In classical logic, $\neg P$ is defined such that $P \vee \neg P$ holds universally. Negation is the complement of truth. But in OHTT, there is no assumption that $J$ is either coherent or gapped. The open is a genuine third state—not a truth value, but a witnessing state.

\paragraph{Intuitionistic negation.} In intuitionistic logic, $\neg P$ is defined as $P \to \bottom$. A proof of $\neg P$ is a function that transforms any proof of $P$ into a proof of absurdity. This is constructive, but it is still \textit{derivative}: negation is defined in terms of implication and absurdity, not given primitively. In OHTT, gap is primitive. We do not define $\gap{\Gamma}{J}$ as ``$\coh{\Gamma}{J}$ implies $\bottom$.'' We assert it directly.

\paragraph{Paraconsistent approaches.} Some logics allow both $P$ and $\neg P$ to hold without explosion. OHTT is not paraconsistent. The Exclusion Law explicitly forbids simultaneous coherence and gap for the same judgment. What OHTT permits is that \textit{neither} holds—the open—which paraconsistent logics do not typically thematize.

The closest relatives are perhaps \textit{bilattice logics} (Belnap, Ginsberg) which track both truth and information, or \textit{signed logics} with explicit polarity. But OHTT differs in its proof-relevance: we track not just \textit{that} a judgment is coherent or gapped, but \textit{how}—the witness carries structure.

\section{Horns as Compositional Tension}

The judgment forms and the Exclusion Law constitute the entire axiomatic structure of OHTT. From here, we derive the characteristic construction: the horn.

Let $\comp$ denote a compositional relation between judgments. In different instantiations of OHTT, this may be:

\begin{itemize}[leftmargin=2em]
\item \textit{Derivability} (in proof theory): $J \comp K$ if $K$ is derivable from $J$
\item \textit{Substitution} (in type theory): $J \comp K$ if $K$ results from substituting into $J$
\item \textit{Transport} (in HoTT): $J \comp K$ if $K$ is the transported judgment along a path
\item \textit{Semantic entailment}: $J \comp K$ if the meaning of $J$ implies the meaning of $K$
\end{itemize}

The specific choice of $\comp$ depends on the application domain. What matters is that $\comp$ has a compositional character: if $J \comp K$ and $K \comp L$, then we may ask whether $J \comp L$.

\begin{definition}[The Horn]\label{def:horn}
For judgments $J, K, L$ and compositional relation $\comp$, the \keyterm{open horn type} is:
\[
\horn{J}{K}{L} \;\defn\; (\coh{\Gamma}{J \comp K}) \times (\coh{\Gamma}{K \comp L}) \times (\gap{\Gamma}{J \comp L})
\]
An inhabitant of $\horn{J}{K}{L}$ is a triple $(\gamma_1, \gamma_2, \omega)$ where $\gamma_1$ witnesses that $J$ composes to $K$, $\gamma_2$ witnesses that $K$ composes to $L$, and $\omega$ witnesses that the direct composition from $J$ to $L$ is gapped.
\end{definition}

The horn is \textit{data}, not failure. In ordinary category theory or type theory, if we have $f : A \to B$ and $g : B \to C$, the composite $g \circ f : A \to C$ is \textit{assumed to exist}. Composition is total. The horn records the situation where composition is \textit{witnessed blocked}: we have the two steps, but the closure is gapped.

\subsection{Geometric Intuition}

The terminology ``horn'' comes from simplicial homotopy theory. A 2-simplex (triangle) has three vertices and three edges:

\[
\begin{tikzcd}[row sep=2.5em, column sep=2.5em]
& K \arrow[dr, "\gamma_2"] & \\
J \arrow[ur, "\gamma_1"] \arrow[rr, dashed, "\omega"'] & & L
\end{tikzcd}
\]

A \textit{horn} $\hornsimp{2}{1}$ consists of two edges of a potential triangle—here, $\gamma_1 : J \to K$ and $\gamma_2 : K \to L$—without the third edge. In a \textit{Kan complex} (the standard simplicial model for spaces), every horn can be filled: the missing edge can always be supplied. This is the Kan condition.

In OHTT, we drop this assumption. Instead, the third edge may be:

\begin{itemize}[leftmargin=2em]
\item \textbf{Coherently filled}: there exists $\gamma : \coh{\Gamma}{J \comp L}$ completing the triangle
\item \textbf{Gap-witnessed}: there exists $\omega : \gap{\Gamma}{J \comp L}$, and we have a horn $\horn{J}{K}{L}$
\item \textbf{Open}: neither coherence nor gap is witnessed for $J \comp L$
\end{itemize}

The horn $\horn{J}{K}{L}$ is inhabited precisely in the second case: two coherent steps, one gapped closure. This is \textit{structural tension} made formal.

\begin{aside}
A horn in a text is the record of local coherence with global rupture. Two sentences may each cohere with a third, while the direct path between them is witnessed as gapped. Three theories may be pairwise consistent but globally inconsistent. The horn does not merely note that something is missing; it records the specific structure of what coheres and what does not.
\end{aside}

\subsection{Higher Horns}

The construction iterates. Given coherence witnesses $\gamma, \gamma' : \coh{\Gamma}{J}$, we may ask whether $\gamma$ and $\gamma'$ themselves cohere—whether there is a ``path between paths.''

Define the level-1 judgment type as coherence witnesses themselves:
\[
J_1 \;\defn\; \Sum{J}{\mathsf{Judgment}}{(\coh{\Gamma}{J})}
\]

Then we instantiate OHTT again at level 1, with $J_1$ as the judgment type. This gives:

\begin{itemize}[leftmargin=2em]
\item Coherence between coherences: $\eta : \coh{\Gamma}{(\gamma \comp \gamma')}$
\item Gap between coherences: $\nu : \gap{\Gamma}{(\gamma \comp \gamma')}$
\item Level-1 horns: $\hornlevel{1}{\gamma}{\gamma'}{\gamma''}$
\end{itemize}

This is exactly the structure of homotopy type theory's higher identity types—but with gap at every level. In HoTT, paths between paths are homotopies; homotopies between homotopies are 2-homotopies; and so on. In OHTT, at each level we have coherence, gap, and open.

The construction is fractal: the same calculus, applied recursively, generates the full tower of higher structure. This is the sense in which OHTT is a \textit{meta}logic: it provides the scaffolding within which ordinary type theories live as fragments.

\section{The Coherent Fragment}

We now make precise the relationship between OHTT and dependent type theory.

\begin{definition}[The Coherent Fragment]\label{def:coherent-fragment}
The \keyterm{coherent fragment} of OHTT is the subcalculus that uses only the judgment form $\coh{\Gamma}{J}$, never $\gap{\Gamma}{J}$.
\end{definition}

In the coherent fragment:

\begin{itemize}[leftmargin=2em]
\item The Exclusion Law is trivially satisfied (no $\vdash^-$ judgments to conflict)
\item Horns cannot be inhabited (they require a $\vdash^-$ component)
\item All unwitnessed judgments are open (gap is unavailable)
\end{itemize}

\begin{definition}[MLTT as Coherent Fragment]\label{def:mltt-fragment}
Let the judgment forms of Martin-Löf Type Theory be:
\begin{align*}
&\Gamma \ctx && \text{($\Gamma$ is a well-formed context)} \\
&\jdg{\Gamma}{\istype{A}} && \text{($A$ is a type in context $\Gamma$)} \\
&\jdg{\Gamma}{t : A} && \text{($t$ has type $A$ in context $\Gamma$)} \\
&\jdg{\Gamma}{s \defeq t : A} && \text{($s$ and $t$ are definitionally equal)}
\end{align*}
The rules of MLTT ($\Pi$-formation, $\Pi$-introduction, $\Pi$-elimination, $\Sigma$-types, identity types, universes, etc.) are \textit{coherence constructors}: rules that generate $\vdash^+$ judgments from $\vdash^+$ premises.
\end{definition}

This is not an encoding of MLTT into OHTT. It is not a translation or model. It is the observation that \textit{MLTT is the coherent fragment of OHTT}—the part that uses only the right hand.

Working entirely within the coherent fragment, a mathematician derives judgments using the standard rules. They may attempt to inhabit a type and fail, but they cannot \textit{witness} the failure—the failure is absence, not structure. OHTT extends this by making failure structurable.

\subsection{HoTT as Coherent Fragment with Higher Structure}

Homotopy Type Theory extends MLTT with:

\begin{itemize}[leftmargin=2em]
\item \textbf{Univalence}: type equivalence induces identity, $(A \eqv B) \eqv (\Id{\UU}{A}{B})$
\item \textbf{Higher inductive types}: types specified with path constructors, not just point constructors
\item \textbf{The homotopy interpretation}: types are spaces, identity types are path spaces
\end{itemize}

In the simplicial semantics of HoTT, types are interpreted as Kan complexes—simplicial sets where \textit{every horn fills}. This is the Kan condition. Paths (inhabitants of identity types) are 1-simplices; homotopies are 2-simplices; and so on.

OHTT generalizes this. The coherent fragment with higher structure corresponds to HoTT: paths are level-0 coherence witnesses, homotopies are level-1 coherence witnesses. But OHTT allows gap witnesses at each level, breaking the Kan condition.

\begin{center}
\textit{HoTT = Coherent Fragment of OHTT + Higher Structure + Kan Filling}
\end{center}
\[
\textit{OHTT = HoTT $-$ Kan Filling + Gap Witnesses}
\]

\subsection{What OHTT Adds}

OHTT extends the coherent fragment by introducing:

\paragraph{1. Witnessed typing failures.} In MLTT, if we cannot derive $t : A$, we simply lack a derivation. In OHTT, we can witness $\gap{\Gamma}{(t : A)}$—the judgment that $t$ does \textit{not} have type $A$. This is not ``$t$ has some other type'' or ``the derivation fails.'' It is positive evidence of non-typing.

\paragraph{2. Witnessed non-identity.} In HoTT, $\neg(\Id{A}{x}{y})$ is defined as $(\Id{A}{x}{y}) \to \bottom$. This says: any path from $x$ to $y$ leads to absurdity. In OHTT, we can witness $\gap{\Gamma}{(p : \Id{A}{x}{y})}$ directly—the judgment that $x$ and $y$ are gapped. This is stronger: not ``a path would be catastrophic'' but ``we have witnessed non-coherence.''

\paragraph{3. Transport horns.} Given $t : B(x)$ and $p : \Id{A}{x}{y}$, HoTT guarantees $\transport{p}{t} : B(y)$. In OHTT, transport may be gapped: $\gap{\Gamma}{(\transport{p}{t} : B(y))}$. This is the \textit{transport horn}—the central example of Chapter~II—where local coherence (the term exists, the path exists) meets global rupture (the transported term does not cohere).

\paragraph{4. Horn structure at all levels.} Whenever two coherence witnesses compose but their composition is gapped, we have a horn. This structure exists at every level of the type-theoretic hierarchy, recording where local derivability fails to compose globally.

\section{Summary}

The Open Horn Type Theory consists of:

\begin{itemize}[leftmargin=2em]
\item \textbf{Two judgment forms}: $\coh{\Gamma}{J}$ (coherence witness) and $\gap{\Gamma}{J}$ (gap witness)
\item \textbf{One axiom}: Exclusion (coherence and gap for the same judgment yield $\bottom$)
\item \textbf{One derived structure}: The horn $\horn{J}{K}{L}$ recording two coherent steps with a gapped closure
\item \textbf{One relationship}: Dependent type theory is the coherent ($\vdash^+$-only) fragment
\end{itemize}

Everything else—the simplicial semantics of Part~II, the obstruction examples of Part~III, the dynamics and persistence of later chapters—builds on this foundation without extending it. The calculus is complete.

What makes this a new logic is not complexity but \textit{reorientation}. By taking gap as primitive, we gain the ability to structure what does not cohere—to witness rupture, not merely fail to witness coherence. This is a small change with large consequences.

\begin{aside}
What this lets us say that HoTT cannot: ``I have witnessed that these two paths do not cohere—not that their coherence would be catastrophic, but that the space between them is structured absence.'' The open horn is not failure. It is the shape of what does not close.
\end{aside}


\chapter{Transport and the Open Horn}
\label{ch:transport}

\begin{flushright}
\textit{The path exists. The term exists.\\
What does not exist is their meeting.}\\[0.5ex]
{\small — The transport horn}
\end{flushright}

\bigskip

In Chapter~I we gave the abstract calculus: two judgment forms, one exclusion law, and the derived notion of horn. But a calculus without examples is a skeleton without flesh. We now show that the transport horn—the situation where a term and a path both exist, but the transported term does not cohere—is the central construction that distinguishes OHTT from HoTT.

This chapter has three movements. First, we recall transport in ordinary type theory and its simplicial semantics. Second, we define the transport horn and show that it cannot be expressed in HoTT. Third, we give three classes of examples—topological, semantic, and logical—that demonstrate the phenomenon is not artificial but ubiquitous.

\section{Transport in Homotopy Type Theory}

We begin with the standard picture.

\begin{definition}[Type Family]\label{def:type-family}
A \keyterm{type family} over $A$ is a function $B : A \to \UU$ assigning to each $x : A$ a type $B(x)$.
\end{definition}

In the homotopy interpretation, a type family is a fibration: $B$ assigns to each point of the base space $A$ a fiber, and the fibers vary continuously (in an appropriate sense) as we move through $A$.

\begin{definition}[Transport]\label{def:transport}
Given a type family $\Fam{B}{A}$, a term $t : B(x)$, and a path $p : \Id{A}{x}{y}$, the \keyterm{transport} of $t$ along $p$ is a term:
\[
\transport{p}{t} : B(y)
\]
\end{definition}

In HoTT, transport is not merely definable—it is \textit{derivable from the induction principle for identity types}. The J-rule guarantees: for any type family $B$ over $A$, any $x : A$, and any $t : B(x)$, if we have a path $p : \Id{A}{x}{y}$, we can construct a term of type $B(y)$.

This is the fundamental operation that makes dependent types ``work'' in the homotopical setting. Without transport, a term in $B(x)$ would be trapped at $x$, unable to move along paths. With transport, terms flow through the type family as we traverse the base.

\begin{aside}
The power of transport is also its danger. In HoTT, transport \textit{always succeeds}. Given any path, any term moves. But this universality obscures situations where transport \textit{should} fail—where the path exists, the term exists, but their combination does not cohere. OHTT makes such failures visible.
\end{aside}

\subsection{The Simplicial Picture}

In the simplicial model of HoTT, types are Kan complexes and type families are Kan fibrations. Let us recall what transport looks like geometrically.

\begin{definition}[Kan Fibration]\label{def:kan-fibration}
A map $p : E \to B$ of simplicial sets is a \keyterm{Kan fibration} if it has the right lifting property against all horn inclusions. That is, for every horn $h : \hornsimp{n}{k} \to E$ and every extension $\bar{h} : \simp{n} \to B$ of $p \circ h$, there exists a lift $\tilde{h} : \simp{n} \to E$ such that $p \circ \tilde{h} = \bar{h}$ and $\tilde{h}$ extends $h$.
\end{definition}

In less technical terms: if we have a partial simplex in the total space $E$ (a horn), and we know how the full simplex should look downstairs in $B$, we can always complete the simplex upstairs.

Now consider transport. We have:
\begin{itemize}[leftmargin=2em]
\item A point $t$ in the fiber $E_x$ over $x \in B$
\item A path $\gamma : \simp{1} \to B$ from $x$ to $y$
\end{itemize}

We want to lift $\gamma$ to a path in $E$ starting at $t$. This is exactly a horn-filling problem:

\[
\begin{tikzcd}[row sep=2em, column sep=3em]
\{0\} \arrow[r, "t"] \arrow[d, hook] & E \arrow[d, "p"] \\
\simp{1} \arrow[r, "\gamma"'] \arrow[ur, dashed, "\tilde{\gamma}"] & B
\end{tikzcd}
\]

The horn $\hornsimp{1}{1}$ consists of just the vertex $\{0\}$; we are asking to extend it to a full 1-simplex. The Kan fibration condition guarantees a lift $\tilde{\gamma}$ exists. The transported term is:
\[
\transport{\gamma}{t} \;=\; \tilde{\gamma}(1)
\]

\begin{proposition}[Transport as Horn-Filling]\label{prop:transport-horn-filling}
In the simplicial model, transport along a path $\gamma : \Id{A}{x}{y}$ corresponds to filling the horn $\hornsimp{1}{1}$ in the total space of the fibration.
\end{proposition}

This is the key insight: \textit{transport is horn-filling}. The Kan condition guarantees all horns fill, hence transport always succeeds. Remove the Kan condition, and transport may fail.

\section{The Transport Horn}

We now define the central construction of OHTT.

\begin{definition}[Transport Horn]\label{def:transport-horn}
Given a type family $\Fam{B}{A}$, a term $t : B(x)$, a path $p : \Id{A}{x}{y}$, and a context $\Gamma$, the \keyterm{transport horn} is an inhabitant of the type:
\begin{align*}
\Lambda_{\mathsf{tr}}(B, t, p) \;\defn\;\; 
&\coh{\Gamma}{(t : B(x))} \\
\times\; &\coh{\Gamma}{(p : \Id{A}{x}{y})} \\
\times\; &\gap{\Gamma}{(\transport{p}{t} : B(y))}
\end{align*}
\end{definition}

An inhabitant $(\gamma_t, \gamma_p, \omega)$ of the transport horn witnesses:
\begin{itemize}[leftmargin=2em]
\item $\gamma_t$: the term $t$ coherently inhabits $B(x)$
\item $\gamma_p$: the path $p$ coherently connects $x$ to $y$
\item $\omega$: the transported term $\transport{p}{t}$ is \textit{gapped} at $B(y)$
\end{itemize}

This is precisely the situation that HoTT cannot express. In HoTT, given $\gamma_t$ and $\gamma_p$, the term $\transport{p}{t} : B(y)$ is \textit{derivable}. There is no way to assert that transport fails; the rules guarantee it succeeds.

In OHTT, we can have coherent premises and a gapped conclusion. The horn records the structural tension: local coherence (the term is well-typed, the path is well-formed) with global rupture (their combination does not cohere).

\subsection{The Geometry of Transport Failure}

What does a transport horn look like simplicially?

In the ruptured simplicial model (Definition~\ref{def:ruptured-kan} below), we have a map $p : E \to B$ that is \textit{not} required to be a Kan fibration. Given:
\begin{itemize}[leftmargin=2em]
\item A 0-simplex $t \in E_x$ (a point in the fiber over $x$)
\item A 1-simplex $\gamma : x \to y$ in $B$ (a path in the base)
\end{itemize}

We form the lifting problem as before. But now three outcomes are possible:

\begin{enumerate}[leftmargin=2em]
\item \textbf{Coherent lift}: There exists $\tilde{\gamma} : \simp{1} \to E$ lifting $\gamma$ with $\tilde{\gamma}(0) = t$. Transport succeeds; we have $\transport{\gamma}{t} = \tilde{\gamma}(1)$.

\item \textbf{Gapped lift}: We witness that no such lift exists. The horn $(\{t\}, \gamma)$ is gap-witnessed; transport is blocked.

\item \textbf{Open}: Neither coherence nor gap is witnessed. The transport question is undetermined.
\end{enumerate}

The transport horn is case (2): the data $t$ and $\gamma$ are coherent, but their combination—the lift—is gapped.

\[
\begin{tikzcd}[row sep=2em, column sep=3em]
\{0\} \arrow[r, "t"] \arrow[d, hook] & E \arrow[d, "p"] \\
\simp{1} \arrow[r, "\gamma"'] \arrow[ur, dashed, "\nexists" description, crossing over] & B
\end{tikzcd}
\]

The dashed arrow with $\nexists$ indicates: not merely ``we haven't found a lift'' but ``we have witnessed that no lift coheres.''

\subsection{Relation to the General Horn}

The transport horn is an instance of the general horn schema from Chapter~I. Recall:
\[
\horn{J}{K}{L} \;\defn\; (\coh{\Gamma}{J \comp K}) \times (\coh{\Gamma}{K \comp L}) \times (\gap{\Gamma}{J \comp L})
\]

For transport, instantiate:
\begin{align*}
J &\;\defn\; (t : B(x)) \\
K &\;\defn\; (p : \Id{A}{x}{y}) \\
L &\;\defn\; (\transport{p}{t} : B(y)) \\
J \comp K &\;\defn\; \text{``$t$ and $p$ are composable data for transport''} \\
K \comp L &\;\defn\; \text{``$p$ determines the target type $B(y)$ for transport''} \\
J \comp L &\;\defn\; \text{``$t$ transports along $p$ to inhabit $B(y)$''}
\end{align*}

The transport horn asserts: the first two compositions cohere (we have the data), but the third is gapped (the transport itself fails).

\section{The Ruptured Simplicial Model}

To give semantics to OHTT, we generalize the simplicial model of HoTT.

\begin{definition}[Ruptured Simplicial Set]\label{def:ruptured-sset}
A \keyterm{ruptured simplicial set} is a tuple $(X, \Coh, \Gap, \Excl)$ where:
\begin{itemize}[leftmargin=2em]
\item $X$ is a graded set $X = (X_0, X_1, X_2, \ldots)$ with face and degeneracy maps
\item $\Coh_n \subseteq X_n$ is the set of \keyterm{coherently witnessed} $n$-simplices
\item $\Gap_n \subseteq \{\text{horns in } X\}$ is the set of \keyterm{gap-witnessed} horns
\item $\Excl$: if $h \in \Gap_n$, then no $\sigma \in \Coh_n$ fills $h$
\end{itemize}
\end{definition}

The key difference from ordinary simplicial sets: we track not just which simplices exist, but which are \textit{witnessed} as coherent, and which horns are \textit{witnessed} as unfillable.

\begin{definition}[Ruptured Kan Complex]\label{def:ruptured-kan}
A ruptured simplicial set $(X, \Coh, \Gap, \Excl)$ is a \keyterm{ruptured Kan complex} if for every horn $h$ in $X$, exactly one of the following holds:
\begin{enumerate}[leftmargin=2em]
\item There exists $\sigma \in \Coh$ filling $h$ (the horn is coherently filled)
\item $h \in \Gap$ (the horn is gap-witnessed)
\item Neither (the horn is open)
\end{enumerate}
\end{definition}

An ordinary Kan complex is a ruptured Kan complex where $\Gap = \emptyset$ and every horn has a coherent filler. HoTT lives in this special case. OHTT lives in the general case.

\begin{definition}[Ruptured Kan Fibration]\label{def:ruptured-fibration}
A map $p : E \to B$ of ruptured simplicial sets is a \keyterm{ruptured Kan fibration} if the lifting property holds \textit{up to the gap structure}: for every lifting problem, either a coherent lift exists, or the lifting horn is gap-witnessed, or the problem is open.
\end{definition}

\begin{proposition}[Transport in Ruptured Fibrations]\label{prop:ruptured-transport}
In a ruptured Kan fibration $p : E \to B$, given $t \in E_x$ and $\gamma : x \to y$ in $B$, the transport of $t$ along $\gamma$ is:
\begin{enumerate}[leftmargin=2em]
\item \textbf{Coherent} if the lifting horn has a coherent filler
\item \textbf{Gapped} if the lifting horn is gap-witnessed
\item \textbf{Open} if the lifting horn is neither
\end{enumerate}
\end{proposition}

This proposition is the semantic content of the transport horn. OHTT allows case (2); HoTT forbids it.

\section{Three Classes of Examples}

We now give concrete examples of transport horns in three domains: topology, semantics, and logic. These are previews of the detailed case studies in Part~III.

\subsection{Topological Obstruction: Monodromy}

Let $p : E \to B$ be a covering space (a particularly nice kind of fibration). Choose a base point $x \in B$ and consider a loop $\gamma : x \to x$ in $B$.

If $t \in E_x$ is a point in the fiber, we can try to lift $\gamma$ to a path in $E$ starting at $t$. For covering spaces, the lift exists and is unique. But the lifted path may not return to $t$—it may end at a different point $t' \in E_x$.

The \keyterm{monodromy} of the covering is the permutation of the fiber induced by lifting loops.

\begin{example}[Möbius Strip]\label{ex:mobius}
Let $B = S^1$ (the circle) and let $E$ be the boundary of a Möbius strip—a double cover of $S^1$. The fiber over any point has two elements: call them $+$ and $-$.

The intuition is immediate: walk once around the base circle, starting at a point $b$. Downstairs, you return to $b$—the loop closes. But upstairs, if you start at $+$ and walk, you arrive at $-$. \textit{You come back changed.}

Mathematically: the lift of the loop exists and is unique (it's the path starting at $+$). But it is not a loop based at $+$—it ends at $-$. The question ``does this loop lift to a loop at my starting point?'' has a definite negative answer, and the monodromy (the sheet-swap) is the certificate.
\end{example}

This is a transport horn:
\begin{itemize}[leftmargin=2em]
\item The starting point $+ \in E_b$ is coherent (it exists)
\item The loop $\gamma : b \to b$ is coherent (it exists)
\item The based lift is gapped (it does not exist as a loop at $+$)
\end{itemize}

The gap witness is the monodromy: not silence about whether a lift exists, but positive data recording \textit{how} the lift fails to close. The deck transformation that swaps $+$ and $-$ is the obstruction made structural.

\begin{aside}
The Möbius strip is the simplest non-trivial example of this phenomenon. The twist is not just a curiosity—it is a gap witness. OHTT lets us say: walking around the base is coherent, your starting point is coherent, but returning to where you started is gapped. The formalism holds what the image shows.
\end{aside}

\subsection{Semantic Obstruction: Polysemy}

Consider natural language as a type-theoretic structure. Let:
\begin{itemize}[leftmargin=2em]
\item $A = $ the type of word tokens (occurrences of words in context)
\item $B(w) = $ the type of meanings (senses) of word $w$
\end{itemize}

Two tokens of the ``same word'' are related by a path in $A$. The word ``bank'' in a financial context and ``bank'' in a river context are connected by lexical identity—they are the same lexeme.

Let $x = $ ``bank'' (financial) and $y = $ ``bank'' (river). Let $p : \Id{A}{x}{y}$ witness their lexical identity.

Let $t : B(x)$ be the meaning ``financial institution.''

What is $\transport{p}{t}$? It would be: the meaning of ``bank'' (river) obtained by transporting the financial meaning along the lexical identity path.

But this doesn't cohere. The river sense of ``bank'' is \textit{not} ``financial institution.'' The lexical identity exists, but it does not transport meanings.

\begin{example}[Polysemy Horn]\label{ex:polysemy}
\[
\Lambda_{\mathsf{tr}}(B, \text{``financial institution''}, p_{\text{lexical}})
\]
is inhabited by:
\begin{itemize}[leftmargin=2em]
\item $\gamma_t$: witnessing that ``financial institution'' is a coherent meaning of ``bank'' (financial)
\item $\gamma_p$: witnessing that ``bank'' (financial) and ``bank'' (river) are lexically identical
\item $\omega$: witnessing that ``financial institution'' is \textit{not} a coherent meaning of ``bank'' (river)
\end{itemize}
\end{example}

In HoTT, if we had a path $p : x = y$, transport would force meanings to transfer. We would be unable to model polysemy—the fact that the same word has genuinely different meanings that do not reduce to each other.

In OHTT, the transport horn records exactly this: lexical identity without semantic identity. The path exists. The meaning exists. But the meaning does not travel along the path.

\begin{aside}
A natural question: how would one actually \textit{measure} whether meanings cohere? Distributional semantics provides an answer. In word embedding models—and in the attention mechanisms of large language models—words live in high-dimensional vector spaces where semantic similarity becomes geometric proximity. ``Bank'' (financial) and ``bank'' (river) may share a lexical token, but their contextual embeddings land in distant regions of meaning-space. Cosine distance, contextual clustering, and attention patterns give us \textit{computable} measures of semantic (dis)continuity.

The transport horn formalizes this divergence: lexical proximity does not imply semantic coherence. In an empirical instantiation of OHTT, the gap witness could be precisely this measurable distance—the geometric evidence that meaning does not travel along the lexical identity path.
\end{aside}

\subsection{Logical Obstruction: Underivability}

Consider a dependent type theory as a category with families. Let:
\begin{itemize}[leftmargin=2em]
\item $A = $ the type of contexts
\item $B(\Gamma) = $ the type of judgments derivable in context $\Gamma$
\end{itemize}

A context morphism $\sigma : \Delta \to \Gamma$ (a substitution) induces a path $p_\sigma : \Id{A}{\Gamma}{\Delta}$ in an appropriate sense.

Given a judgment $J$ derivable in $\Gamma$, we can ask: is $J[\sigma]$ (the substituted judgment) derivable in $\Delta$?

In standard type theory, substitution preserves derivability. If $\jdg{\Gamma}{t : T}$, then $\jdg{\Delta}{t[\sigma] : T[\sigma]}$. This is the \textit{substitution lemma}.

But consider a richer setting where derivability is resource-sensitive, or where contexts carry additional structure. Then substitution may fail to preserve derivability.

\begin{example}[Resource Sensitivity]\label{ex:resource}
Let $\Gamma = x : \mathsf{File}, y : \mathsf{ReadPerm}(x)$ and let $J = \mathsf{read}(x) : \mathsf{Data}$.

Let $\sigma : \Delta \to \Gamma$ be a substitution where $\Delta$ provides the file but not the read permission.

Then:
\begin{itemize}[leftmargin=2em]
\item $\coh{}{\jdg{\Gamma}{J}}$: the read is derivable with permission
\item $\coh{}{(\sigma : \Delta \to \Gamma)}$: the substitution is well-formed
\item $\gap{}{\jdg{\Delta}{J[\sigma]}}$: the read is not derivable without permission
\end{itemize}

This is a transport horn: the judgment is derivable, the substitution exists, but the substituted judgment is gapped.
\end{example}

The gap witness here is not merely ``no derivation found'' but a \textit{certificate of underivability}—a proof that the resources required for $J$ are not available in $\Delta$.

\begin{aside}
This connects OHTT to certified compilation, proof-carrying code, and information flow analysis. Systems that track permissions, capabilities, or security levels routinely encounter situations where operations are well-typed in one context but blocked in another. The transport horn is the natural formalization: substitution exists, but derivability does not transport.
\end{aside}

\section{The Functoriality Question}

In HoTT, transport is functorial. Given paths $p : \Id{A}{x}{y}$ and $q : \Id{A}{y}{z}$, we have:
\[
\transport{q \cdot p}{t} = \transport{q}{\transport{p}{t}}
\]

Transport along a composite path equals iterated transport.

In OHTT, this equation becomes a \textit{question}. We may have:
\begin{itemize}[leftmargin=2em]
\item $\transport{p}{t}$ coherent
\item $\transport{q}{\transport{p}{t}}$ coherent
\item But $\transport{q \cdot p}{t}$ gapped—or open
\end{itemize}

The composite path may fail to transport even when the individual segments succeed.

\begin{example}[Cumulative Obstruction]\label{ex:cumulative}
In the polysemy example, consider three senses of ``crane'': the bird, the machine, and the poetic verb (to crane one's neck).

\begin{itemize}[leftmargin=2em]
\item Path $p$: bird $\to$ machine (both are tall, thin, with a hook-like part)
\item Path $q$: machine $\to$ verb (the machine's motion resembles necking craning)
\end{itemize}

Some semantic features transport along $p$ (tallness, thinness). Some transport along $q$ (extension, reaching). But transporting directly from bird to verb via $q \cdot p$ may lose all coherent features—the composite obstruction is worse than the sum of parts.
\end{example}

This failure of functoriality is not a bug but a feature. It captures the phenomenon of \textit{semantic drift}: meanings can travel stepwise through intermediate senses but lose coherence over longer paths.

\begin{definition}[Functoriality Horn]\label{def:functoriality-horn}
Given paths $p : \Id{A}{x}{y}$ and $q : \Id{A}{y}{z}$, and a term $t : B(x)$, the \keyterm{functoriality horn} is:
\begin{align*}
\Lambda_{\mathsf{func}}(t, p, q) \;\defn\;\; 
&\coh{\Gamma}{(\transport{p}{t} : B(y))} \\
\times\; &\coh{\Gamma}{(\transport{q}{(\transport{p}{t})} : B(z))} \\
\times\; &\gap{\Gamma}{(\transport{q \cdot p}{t} : B(z))}
\end{align*}
\end{definition}

The functoriality horn witnesses: stepwise transport succeeds, but direct transport along the composite fails. This is a higher-order obstruction invisible to HoTT.

\section{What This Lets Us Say}

Let us be explicit about the expressive gain.

\paragraph{In HoTT, we can say:}
\begin{itemize}[leftmargin=2em]
\item ``$t$ has type $B(x)$''
\item ``$p$ is a path from $x$ to $y$''
\item ``$\transport{p}{t}$ has type $B(y)$'' (automatically, by the rules)
\item ``$\transport{p}{t} = s$'' for some specific $s : B(y)$
\end{itemize}

\paragraph{In HoTT, we cannot say:}
\begin{itemize}[leftmargin=2em]
\item ``Transport of $t$ along $p$ is blocked''
\item ``There is an obstruction to transport'' (as positive data)
\item ``The path exists but does not carry this term''
\end{itemize}

\paragraph{In OHTT, we can additionally say:}
\begin{itemize}[leftmargin=2em]
\item $\gap{\Gamma}{(\transport{p}{t} : B(y))}$ — ``Transport is witnessed as blocked''
\item $\Lambda_{\mathsf{tr}}(B, t, p)$ inhabited — ``We have a transport horn''
\item The obstruction $\omega$ carries structure—\textit{why} transport fails
\end{itemize}

This is the fundamental extension. HoTT can express what transports; OHTT can also express what does not—and \textit{why}.

\section{Summary}

The transport horn is the keystone of OHTT:

\begin{itemize}[leftmargin=2em]
\item \textbf{Definition}: $\Lambda_{\mathsf{tr}}(B, t, p) = \coh{}{(t : B(x))} \times \coh{}{(p : x = y)} \times \gap{}{(\transport{p}{t} : B(y))}$

\item \textbf{Geometry}: Transport is horn-filling. The transport horn is a horn that is witnessed unfillable.

\item \textbf{Semantics}: Ruptured Kan fibrations allow lifting problems to be coherent, gapped, or open.

\item \textbf{Examples}: Monodromy in covering spaces (topological), polysemy in language (semantic), resource-sensitive derivability (logical).

\item \textbf{Higher structure}: Functoriality may fail; composite paths may be gapped even when segments transport.
\end{itemize}

The transport horn is not an exotic edge case. It is the generic situation whenever paths (identities, equivalences, substitutions) exist between points whose internal structure does not respect those paths. HoTT assumes this never happens. OHTT allows it, and gives it structure.

\begin{aside}
What this lets us say that HoTT cannot: ``The word is the same, but the meaning does not follow. The path is there, but the term cannot walk it. The identity exists—and yet, something refuses to be identified.'' The transport horn is the formal shape of that refusal.
\end{aside}

\part{The Geometric Model}
\label{part:model}


\chapter{Ruptured Simplicial Sets}
\label{ch:ruptured-sset}

\begin{flushright}
\textit{A space is not just what is there.\\
It is also what is witnessed absent.}\\[0.5ex]
{\small — The ruptured view}
\end{flushright}

\bigskip

In Part~I we gave OHTT as a calculus: judgment forms, the exclusion law, horns. In Chapter~\ref{ch:transport} we sketched how this calculus finds semantics in a modification of the simplicial model. Now we develop that model properly.

This chapter constructs ruptured simplicial sets—the basic objects of our geometry. We begin with a rapid review of ordinary simplicial sets and Kan complexes, then introduce the ruptured variant, and finally characterize the relationship between ruptured and ordinary structures.

\section{Simplicial Sets: A Review}

We assume familiarity with basic category theory but recall the essential definitions.

\begin{definition}[The Simplex Category]\label{def:simplex-cat}
The \keyterm{simplex category} $\Delta$ has:
\begin{itemize}[leftmargin=2em]
\item Objects: finite non-empty ordinals $[n] = \{0, 1, \ldots, n\}$ for $n \geq 0$
\item Morphisms: order-preserving maps $[m] \to [n]$
\end{itemize}
\end{definition}

The morphisms of $\Delta$ are generated by two families:

\begin{itemize}[leftmargin=2em]
\item \textbf{Face maps} $\delta_i : [n-1] \to [n]$ for $0 \leq i \leq n$: the injection that skips $i$
\item \textbf{Degeneracy maps} $\sigma_i : [n+1] \to [n]$ for $0 \leq i \leq n$: the surjection that repeats $i$
\end{itemize}

These satisfy the \textit{simplicial identities}, which we omit but which ensure coherent composition.

\begin{definition}[Simplicial Set]\label{def:sset}
A \keyterm{simplicial set} is a functor $X : \Delta^{\mathrm{op}} \to \mathbf{Set}$. Explicitly, $X$ consists of:
\begin{itemize}[leftmargin=2em]
\item Sets $X_n$ for each $n \geq 0$ (the \keyterm{$n$-simplices})
\item Face maps $d_i : X_n \to X_{n-1}$ for $0 \leq i \leq n$
\item Degeneracy maps $s_i : X_n \to X_{n+1}$ for $0 \leq i \leq n$
\end{itemize}
satisfying the simplicial identities.
\end{definition}

We write $\mathbf{sSet}$ for the category of simplicial sets with natural transformations as morphisms.

\begin{example}[Standard Simplices]\label{ex:standard-simplex}
The \keyterm{standard $n$-simplex} $\simp{n}$ is the representable functor $\Delta(-, [n])$. Concretely:
\begin{itemize}[leftmargin=2em]
\item $\simp{0}$ is a single point
\item $\simp{1}$ is an edge with two vertices
\item $\simp{2}$ is a filled triangle with three vertices, three edges, and one 2-cell
\end{itemize}
\end{example}

\subsection{Horns}

\begin{definition}[Horn]\label{def:horn-simplicial}
The \keyterm{$(n,k)$-horn} $\hornsimp{n}{k}$ for $0 \leq k \leq n$ is the simplicial subset of $\simp{n}$ consisting of all faces \textit{except} the $k$-th face. It is the ``boundary of $\simp{n}$ with one face missing.''
\end{definition}

Concretely:
\begin{itemize}[leftmargin=2em]
\item $\hornsimp{1}{0}$ is the single vertex $\{1\}$ (the edge $\simp{1}$ minus the face opposite vertex 0)
\item $\hornsimp{1}{1}$ is the single vertex $\{0\}$
\item $\hornsimp{2}{0}$ is two edges sharing vertex 0: the edges $[0,1]$ and $[0,2]$
\item $\hornsimp{2}{1}$ is two edges sharing vertex 1: the edges $[0,1]$ and $[1,2]$
\item $\hornsimp{2}{2}$ is two edges sharing vertex 2: the edges $[0,2]$ and $[1,2]$
\end{itemize}

A horn is a partial simplex—all the data needed to potentially form a full simplex, except one piece.

\begin{definition}[Horn Inclusion]\label{def:horn-inclusion}
The \keyterm{horn inclusion} is the canonical map $\iota : \hornsimp{n}{k} \hookrightarrow \simp{n}$.
\end{definition}

\subsection{Kan Complexes}

\begin{definition}[Kan Complex]\label{def:kan-complex}
A simplicial set $X$ is a \keyterm{Kan complex} if every horn in $X$ can be filled. Precisely: for every $n \geq 1$, every $0 \leq k \leq n$, and every map $h : \hornsimp{n}{k} \to X$, there exists an extension $\bar{h} : \simp{n} \to X$ such that $\bar{h} \circ \iota = h$.
\[
\begin{tikzcd}[row sep=2em, column sep=2.5em]
\hornsimp{n}{k} \arrow[r, "h"] \arrow[d, hook, "\iota"'] & X \\
\simp{n} \arrow[ur, dashed, "\bar{h}"'] &
\end{tikzcd}
\]
\end{definition}

The Kan condition says: if you have all faces of a potential simplex except one, you can always complete it. This is the defining property that makes Kan complexes behave like spaces.

\begin{theorem}[Kan Complexes as $\infty$-Groupoids]\label{thm:kan-groupoid}
Kan complexes are precisely the $\infty$-groupoids: higher categorical structures where all morphisms at all levels are invertible (up to coherent higher morphisms).
\end{theorem}

This is the homotopy hypothesis (Grothendieck): spaces = $\infty$-groupoids = Kan complexes. In HoTT, types are interpreted as Kan complexes, which is why all paths are invertible and all higher structure is groupoidal.

\section{Ruptured Simplicial Sets}

We now introduce the central construction.

\begin{definition}[Ruptured Simplicial Set]\label{def:ruptured-sset-formal}
A \keyterm{ruptured simplicial set} is a tuple $\mathcal{X} = (X, \Coh, \Gap)$ where:
\begin{enumerate}[leftmargin=2em]
\item $X$ is a simplicial set (the \keyterm{underlying simplicial set})
\item $\Coh = (\Coh_n)_{n \geq 0}$ where $\Coh_n \subseteq X_n$ is the set of \keyterm{coherently witnessed $n$-simplices}
\item $\Gap = (\Gap_{n,k})_{n \geq 1, 0 \leq k \leq n}$ where $\Gap_{n,k} \subseteq \mathrm{Hom}(\hornsimp{n}{k}, X)$ is the set of \keyterm{gap-witnessed $(n,k)$-horns}
\end{enumerate}
subject to the \keyterm{exclusion condition}: if $h \in \Gap_{n,k}$, then there is no $\sigma \in \Coh_n$ such that $\sigma$ fills $h$ (i.e., no $\sigma$ with $d_i(\sigma) = h_i$ for all $i \neq k$).
\end{definition}

\begin{remark}
The exclusion condition is the semantic counterpart of the Exclusion Law from Chapter~I. Coherent filling and gap witnessing cannot coexist for the same horn.
\end{remark}

\begin{definition}[Open Horns]\label{def:open-horn}
A horn $h : \hornsimp{n}{k} \to X$ in a ruptured simplicial set $\mathcal{X}$ is \keyterm{open} if:
\begin{itemize}[leftmargin=2em]
\item There is no $\sigma \in \Coh_n$ filling $h$, and
\item $h \notin \Gap_{n,k}$
\end{itemize}
\end{definition}

Thus every horn in a ruptured simplicial set is in exactly one of three states:
\begin{enumerate}[leftmargin=2em]
\item \textbf{Coherently filled}: admits a filler in $\Coh$
\item \textbf{Gap-witnessed}: belongs to $\Gap$
\item \textbf{Open}: neither
\end{enumerate}

This trichotomy is the geometric realization of the OHTT judgment trichotomy: $\vdash^+$, $\vdash^-$, or unwitnessed.

\subsection{Morphisms}

\begin{definition}[Morphism of Ruptured Simplicial Sets]\label{def:ruptured-morphism}
A \keyterm{morphism} $f : \mathcal{X} \to \mathcal{Y}$ of ruptured simplicial sets is a simplicial map $f : X \to Y$ such that:
\begin{enumerate}[leftmargin=2em]
\item $f$ preserves coherence: if $\sigma \in \Coh^X_n$, then $f(\sigma) \in \Coh^Y_n$
\item $f$ preserves gaps: if $h \in \Gap^X_{n,k}$, then $f \circ h \in \Gap^Y_{n,k}$
\end{enumerate}
\end{definition}

\begin{proposition}\label{prop:ruptured-category}
Ruptured simplicial sets and their morphisms form a category $\mathbf{rsSet}$.
\end{proposition}

\begin{proof}
Identity and composition are inherited from $\mathbf{sSet}$; preservation of $\Coh$ and $\Gap$ is immediate.
\end{proof}

\subsection{The Forgetful Functor}

\begin{definition}[Underlying Simplicial Set]\label{def:underlying}
The \keyterm{forgetful functor} $U : \mathbf{rsSet} \to \mathbf{sSet}$ sends $(X, \Coh, \Gap)$ to $X$.
\end{definition}

This functor forgets the witnessing structure. A ruptured simplicial set carries strictly more information than its underlying simplicial set.

\begin{proposition}[Fibers of the Forgetful Functor]\label{prop:fiber}
For a fixed simplicial set $X$, the fiber $U^{-1}(X)$ is the set of all ways to assign coherence and gap witnessing to simplices and horns of $X$, subject to exclusion.
\end{proposition}

This fiber is large: a single simplicial set admits many ruptured structures. The choice of ruptured structure encodes ``what we know'' about the space—which simplices we have witnessed as coherent, which horns we have witnessed as unfillable.

\section{Ruptured Kan Complexes}

\begin{definition}[Ruptured Kan Complex]\label{def:ruptured-kan-formal}
A ruptured simplicial set $\mathcal{X} = (X, \Coh, \Gap)$ is a \keyterm{ruptured Kan complex} if every horn is either coherently filled, gap-witnessed, or open. (This is automatic from the definitions—the content is that we impose no further Kan-like filling requirement.)
\end{definition}

Equivalently: a ruptured Kan complex is a ruptured simplicial set where we have made no commitment that non-gapped horns must have coherent fillers.

\begin{proposition}[Kan Complexes as Special Ruptured Kan Complexes]\label{prop:kan-as-ruptured}
An ordinary Kan complex $X$ can be viewed as a ruptured Kan complex $(X, X_\bullet, \emptyset)$ where:
\begin{itemize}[leftmargin=2em]
\item $\Coh_n = X_n$ (all simplices are coherently witnessed)
\item $\Gap_{n,k} = \emptyset$ (no horns are gap-witnessed)
\end{itemize}
In this case, every horn is coherently filled (none are gapped or open).
\end{proposition}

This embedding $\mathbf{Kan} \hookrightarrow \mathbf{rsKan}$ exhibits Kan complexes as the ``fully coherent'' ruptured Kan complexes—those with maximal $\Coh$ and empty $\Gap$.

\begin{definition}[Fully Gapped Complex]\label{def:fully-gapped}
A ruptured simplicial set is \keyterm{fully gapped} if every non-degenerate horn is gap-witnessed. This is the opposite extreme from Kan.
\end{definition}

Most interesting ruptured Kan complexes lie between these extremes: some horns fill, some are gapped, some remain open.

\subsection{The Coherent Core}

\begin{definition}[Coherent Core]\label{def:coherent-core}
For a ruptured simplicial set $\mathcal{X} = (X, \Coh, \Gap)$, the \keyterm{coherent core} $X^{\Coh} \subseteq X$ is the smallest simplicial subset containing $\Coh_n$ for all $n$. Equivalently, $X^{\Coh}$ is obtained by closing $\Coh$ under all face and degeneracy maps.
\end{definition}

The coherent core contains exactly what we have witnessed as coherent, together with whatever simplices are forced by the simplicial identities (faces and degeneracies of coherent simplices).

\begin{proposition}\label{prop:core-kan}
If $\mathcal{X}$ is a ruptured Kan complex, then $X^{\Coh}$ is a (possibly partial) Kan complex in the sense that horns within $X^{\Coh}$ that have coherent fillers in $\mathcal{X}$ have fillers in $X^{\Coh}$.
\end{proposition}

\begin{aside}
The coherent core is the ``HoTT fragment'' of a ruptured structure: the part where ordinary type-theoretic reasoning applies. Gaps live outside the core, in the complement. The boundary between core and complement is where transport horns live.
\end{aside}

\section{Gap Witnesses as Structure}

A crucial feature of ruptured simplicial sets: gap witnesses carry structure.

\begin{definition}[Gap Witness]\label{def:gap-witness}
For a horn $h \in \Gap_{n,k}$, a \keyterm{gap witness} for $h$ is the data that $h \in \Gap_{n,k}$—the fact that $h$ has been marked as gap-witnessed.
\end{definition}

In the minimal formulation, $\Gap_{n,k}$ is just a subset, and ``being gap-witnessed'' is a property. But we can enrich this.

\begin{definition}[Typed Gap Structure]\label{def:typed-gap}
A \keyterm{typed ruptured simplicial set} is a tuple $(X, \Coh, \Gap, \Omega)$ where:
\begin{itemize}[leftmargin=2em]
\item $(X, \Coh, \Gap)$ is a ruptured simplicial set
\item $\Omega : \mathrm{dom}(\Gap) \to \mathbf{Set}$ assigns to each gap-witnessed horn $h \in \Gap_{n,k}$ a set $\Omega(h)$ of \keyterm{gap modes}
\end{itemize}
\end{definition}

Gap modes capture \textit{how} a horn fails to fill. Different modes may represent:
\begin{itemize}[leftmargin=2em]
\item Different obstructions (topological, semantic, logical)
\item Different degrees of failure
\item Different witnesses or certificates of non-fillability
\end{itemize}

\begin{example}[Monodromy Modes]\label{ex:monodromy-modes}
In the Möbius strip example from Chapter~\ref{ch:transport}, the gap mode for a transport horn could record which deck transformation obstructs the lift. For the double cover of $S^1$, there are two deck transformations: identity and swap. The non-trivial gap mode is the swap.
\end{example}

For the core development, we work with the untyped version where $\Gap_{n,k}$ is simply a subset. Typed gaps are a natural extension when finer structure is needed.

\section{Constructions on Ruptured Simplicial Sets}

We briefly note that standard simplicial constructions extend to the ruptured setting.

\subsection{Products}

\begin{definition}[Product]\label{def:ruptured-product}
For ruptured simplicial sets $\mathcal{X}$ and $\mathcal{Y}$, the \keyterm{product} $\mathcal{X} \times \mathcal{Y}$ has:
\begin{itemize}[leftmargin=2em]
\item Underlying simplicial set: $X \times Y$ (the ordinary product)
\item $\Coh^{X \times Y}_n = \Coh^X_n \times \Coh^Y_n$
\item $\Gap^{X \times Y}_{n,k} = $ horns whose projection to $X$ is in $\Gap^X$ or whose projection to $Y$ is in $\Gap^Y$
\end{itemize}
\end{definition}

A horn in the product is gapped if either of its projections is gapped. This reflects that obstruction in either factor blocks the combined structure.

\subsection{Mapping Spaces (Outline)}

The notion of mapping space extends to the ruptured setting, but the full development requires care.

\begin{remark}[Mapping Spaces]\label{rem:ruptured-mapping}
For ruptured simplicial sets $\mathcal{X}$ and $\mathcal{Y}$, one expects a \keyterm{mapping space} $\mathrm{Map}(\mathcal{X}, \mathcal{Y})$ where:
\begin{itemize}[leftmargin=2em]
\item $n$-simplices are maps $\mathcal{X} \times \simp{n} \to \mathcal{Y}$ (with appropriate ruptured structure on $\simp{n}$)
\item Coherent $n$-simplices are those maps that preserve coherence
\item Gap-witnessed horns are those whose filling would require a non-coherence-preserving map
\end{itemize}
The details—particularly the canonical ruptured structure on $\simp{n}$ and the verification that $\mathrm{Map}(\mathcal{X}, \mathcal{Y})$ is well-defined in $\mathbf{rsSet}$—are deferred to future work.
\end{remark}

\section{Homotopy in Ruptured Simplicial Sets}

In ordinary simplicial homotopy theory, two maps $f, g : X \to Y$ are homotopic if there exists a map $H : X \times \simp{1} \to Y$ connecting them.

In the ruptured setting, we distinguish:

\begin{definition}[Coherent Homotopy]\label{def:coherent-homotopy}
Let $f, g : \mathcal{X} \to \mathcal{Y}$ be maps of ruptured simplicial sets. A \keyterm{coherent homotopy} from $f$ to $g$ is a homotopy $H$ that:
\begin{itemize}[leftmargin=2em]
\item sends coherent simplices to coherent simplices
\item sends gap-witnessed horns to gap-witnessed or open horns
\end{itemize}
\end{definition}

\begin{remark}[Gapped Homotopy (Informal)]\label{rem:gapped-homotopy}
Informally, a \keyterm{gapped homotopy relation} between $f$ and $g$ is a gap witness for the horn in $\mathrm{Map}(\mathcal{X}, \mathcal{Y})$ whose faces are $f$ and $g$. This depends on the full development of mapping spaces (Remark~\ref{rem:ruptured-mapping}).
\end{remark}

Thus two maps may be:
\begin{itemize}[leftmargin=2em]
\item \textbf{Coherently homotopic}: a coherent homotopy connects them
\item \textbf{Gapped}: we witness that no coherent homotopy connects them
\item \textbf{Open}: their homotopy relation is undetermined
\end{itemize}

This trichotomy at the level of maps mirrors the trichotomy at the level of points (paths) and higher (higher homotopies).

\section{The Trichotomy at Every Level}

We now see the full picture. At every level of the simplicial structure:

\begin{center}
\begin{tabular}{c|c|c|c}
\textbf{Level} & \textbf{Coherent} & \textbf{Gapped} & \textbf{Open} \\
\hline
0 (points) & $\sigma_0 \in \Coh_0$ & — & $\sigma_0 \notin \Coh_0$ \\
1 (paths) & filler in $\Coh_1$ & horn in $\Gap_{1,k}$ & neither \\
2 (homotopies) & filler in $\Coh_2$ & horn in $\Gap_{2,k}$ & neither \\
$n$ & filler in $\Coh_n$ & horn in $\Gap_{n,k}$ & neither \\
\end{tabular}
\end{center}

At level 0, there are no horns, so no gaps—only coherent or not-yet-coherent points. From level 1 upward, the full trichotomy appears.

\begin{theorem}[Interpretation of OHTT Judgments]\label{thm:interpretation}
Let $\mathcal{X}$ be a ruptured Kan complex modeling a context $\Gamma$ in OHTT. Then:
\begin{enumerate}[leftmargin=2em]
\item $\coh{\Gamma}{(p : \Id{A}{x}{y})}$ is interpreted as: the path horn in the identity type has a filler in $\Coh_1$
\item $\gap{\Gamma}{(p : \Id{A}{x}{y})}$ is interpreted as: the path horn is in $\Gap_{1,k}$
\item The judgment is open if neither holds
\end{enumerate}
Similarly for higher identity types at higher levels.
\end{theorem}

This theorem connects the syntactic calculus of Chapter~I to the semantic structures developed here. OHTT judgments \textit{are} statements about ruptured Kan complexes.

\section{Summary}

Ruptured simplicial sets extend simplicial sets with:
\begin{itemize}[leftmargin=2em]
\item \textbf{Coherence marking}: which simplices are witnessed as coherent ($\Coh$)
\item \textbf{Gap witnessing}: which horns are witnessed as unfillable ($\Gap$)
\item \textbf{Exclusion}: coherent filling and gap witnessing are mutually exclusive
\end{itemize}

This gives the trichotomy at every level: coherent, gapped, or open.

Key relationships:
\begin{itemize}[leftmargin=2em]
\item Ordinary Kan complexes embed as ruptured Kan complexes with $\Coh = $ everything and $\Gap = \emptyset$
\item The coherent core $X^{\Coh}$ is the HoTT-like fragment
\item Morphisms preserve both coherence and gaps
\item Homotopy itself becomes three-valued: coherent, gapped, or open
\end{itemize}

In the next chapter, we extend this to fibrations—the structures that model type families and make transport possible (or impossible).

\begin{aside}
What this lets us say that ordinary simplicial theory cannot: ``This space has a hole that is not merely absence of a simplex but witnessed impossibility of one. The obstruction is structure, not lack of structure.'' The ruptured simplicial set is the formal home for obstructions.
\end{aside}


\chapter{Ruptured Fibrations}
\label{ch:ruptured-fib}

\begin{flushright}
\textit{The base has paths. The fiber has points.\\
Whether points travel along paths—that is the question.}\\[0.5ex]
{\small — The fibration problem}
\end{flushright}

\bigskip

A type family $B : A \to \UU$ in HoTT is modeled by a Kan fibration $p : E \to B$. The Kan condition on fibrations guarantees that paths in the base lift to paths in the total space—this is what makes transport work. In OHTT, we drop this guarantee. Paths may lift, or they may be witnessed as non-lifting. This chapter develops the theory of ruptured fibrations.

\section{Kan Fibrations: Review}

We recall the classical notion.

\begin{definition}[Kan Fibration]\label{def:kan-fib-recall}
A map $p : E \to B$ of simplicial sets is a \keyterm{Kan fibration} if it has the right lifting property against all horn inclusions. That is, for every commutative square:
\[
\begin{tikzcd}[row sep=2.5em, column sep=3em]
\hornsimp{n}{k} \arrow[r, "h"] \arrow[d, hook] & E \arrow[d, "p"] \\
\simp{n} \arrow[r, "\sigma"'] \arrow[ur, dashed, "\tilde{\sigma}"] & B
\end{tikzcd}
\]
there exists a diagonal filler $\tilde{\sigma}$ making both triangles commute.
\end{definition}

The lifting property says: if we have a partial simplex upstairs (a horn in $E$) and we know how it should complete downstairs (a full simplex in $B$), then we can complete it upstairs.

\begin{proposition}[Fibers of Kan Fibrations]\label{prop:kan-fibers}
If $p : E \to B$ is a Kan fibration and $b \in B_0$, then the fiber $E_b = p^{-1}(b)$ is a Kan complex.
\end{proposition}

Fibers are ``spaces'' in their own right. The total space $E$ organizes these fiber-spaces into a coherent family over the base $B$.

\begin{proposition}[Path Lifting]\label{prop:path-lifting}
Let $p : E \to B$ be a Kan fibration. Given $e \in E_0$ with $p(e) = b$, and a path $\gamma : b \to b'$ in $B$, there exists a path $\tilde{\gamma} : e \to e'$ in $E$ with $p(\tilde{\gamma}) = \gamma$.
\end{proposition}

This is the 1-dimensional case of the lifting property. The lifted path $\tilde{\gamma}$ may not be unique, but it exists. This existence is what makes transport definable in HoTT.

\section{Ruptured Fibrations}

\begin{definition}[Ruptured Fibration]\label{def:ruptured-fib}
A \keyterm{ruptured fibration} is a map $p : \mathcal{E} \to \mathcal{B}$ of ruptured simplicial sets such that for every lifting problem:
\[
\begin{tikzcd}[row sep=2.5em, column sep=3em]
\hornsimp{n}{k} \arrow[r, "h"] \arrow[d, hook] & E \arrow[d, "p"] \\
\simp{n} \arrow[r, "\sigma"'] & B
\end{tikzcd}
\]
with $h$ coherent (i.e., the faces of $h$ lie in $\Coh^E$) and $\sigma$ coherent (i.e., $\sigma \in \Coh^B_n$), exactly one of the following holds:
\begin{enumerate}[leftmargin=2em]
\item \textbf{Coherent lift}: There exists $\tilde{\sigma} \in \Coh^E_n$ making the diagram commute
\item \textbf{Gapped lift}: The lifting horn $(h, \sigma)$ is in $\Gap^{\mathcal{E}}$
\item \textbf{Open}: Neither coherent lift nor gap witness exists
\end{enumerate}
\end{definition}

The key difference from Kan fibrations: we do not require that coherent lifts always exist. They may, or they may be gapped, or the question may be open.

\begin{remark}
The phrase ``lifting horn $(h, \sigma)$'' means the data of $h$ together with the instruction that it should lift $\sigma$. This is not literally a horn in $E$ until we specify what the missing face should be; rather, it is a \textit{lifting problem} which may or may not have a solution.
\end{remark}

\subsection{Fibers of Ruptured Fibrations}

\begin{definition}[Fiber]\label{def:ruptured-fiber}
For a ruptured fibration $p : \mathcal{E} \to \mathcal{B}$ and a point $b \in \Coh^B_0$, the \keyterm{fiber} $\mathcal{E}_b$ is the ruptured simplicial set with:
\begin{itemize}[leftmargin=2em]
\item Underlying simplicial set: $E_b = p^{-1}(b)$
\item $\Coh^{E_b}_n = \Coh^E_n \cap (E_b)_n$
\item $\Gap^{E_b}_{n,k} = \Gap^E_{n,k}$ restricted to horns lying in $E_b$
\end{itemize}
\end{definition}

\begin{proposition}[Fibers are Ruptured Kan Complexes]\label{prop:ruptured-fiber-kan}
If $p : \mathcal{E} \to \mathcal{B}$ is a ruptured fibration and $b \in \Coh^B_0$, then $\mathcal{E}_b$ is a ruptured Kan complex.
\end{proposition}

\begin{proof}
Horns in $\mathcal{E}_b$ are horns in $\mathcal{E}$ that lie entirely over $b$. Such horns lift the degenerate simplex $s_0^n(b)$ in $B$, which is coherent. By the ruptured fibration condition, each such lifting problem is coherent, gapped, or open—which is the condition for $\mathcal{E}_b$ to be a ruptured Kan complex.
\end{proof}

\subsection{Path Lifting in Ruptured Fibrations}

\begin{proposition}[Ruptured Path Lifting]\label{prop:ruptured-path-lifting}
Let $p : \mathcal{E} \to \mathcal{B}$ be a ruptured fibration. Let $e \in \Coh^E_0$ with $p(e) = b \in \Coh^B_0$, and let $\gamma \in \Coh^B_1$ be a path from $b$ to $b'$. Then exactly one of the following holds:
\begin{enumerate}[leftmargin=2em]
\item \textbf{Coherent lift}: There exists $\tilde{\gamma} \in \Coh^E_1$ with $d_1(\tilde{\gamma}) = e$ and $p(\tilde{\gamma}) = \gamma$
\item \textbf{Gapped lift}: The lifting problem $(e, \gamma)$ is gap-witnessed
\item \textbf{Open}: Neither
\end{enumerate}
\end{proposition}

\begin{proof}
This is the $n=1$, $k=1$ case of the ruptured fibration condition. The horn is $\hornsimp{1}{1} = \{0\}$, and the map $h$ sends $0$ to $e$. The simplex $\sigma$ is $\gamma$. The ruptured fibration condition gives exactly the trichotomy.
\end{proof}

This proposition is the semantic content of the transport horn. Case (1) is coherent transport; case (2) is gapped transport; case (3) is open transport.

\section{Transport in Ruptured Fibrations}

We now make the connection to OHTT transport fully explicit.

\begin{definition}[Transport Operation]\label{def:ruptured-transport}
Let $p : \mathcal{E} \to \mathcal{B}$ be a ruptured fibration modeling a type family $B : A \to \UU$. Let $e \in \Coh^E_0$ be a coherent point in the fiber over $b$, and let $\gamma \in \Coh^B_1$ be a coherent path from $b$ to $b'$.

The \keyterm{transport of $e$ along $\gamma$} is:
\begin{itemize}[leftmargin=2em]
\item If the lifting problem $(e, \gamma)$ is coherently solvable: $\transport{\gamma}{e} = d_0(\tilde{\gamma})$ where $\tilde{\gamma}$ is the coherent lift
\item If the lifting problem is gapped: $\transport{\gamma}{e}$ is undefined; we record $\gap{}{(\transport{\gamma}{e} : \mathcal{E}_{b'})}$
\item If open: $\transport{\gamma}{e}$ is undetermined
\end{itemize}
\end{definition}

In HoTT, transport is a total function: $\mathsf{transport} : (p : x = y) \to B(x) \to B(y)$. In OHTT, transport is a \textit{partial} operation that may succeed, fail (with witness), or be undetermined.

\subsection{The Transport Horn, Semantically}

\begin{theorem}[Semantic Transport Horn]\label{thm:semantic-transport-horn}
In a ruptured fibration $p : \mathcal{E} \to \mathcal{B}$, a transport horn $\Lambda_{\mathsf{tr}}(B, e, \gamma)$ from Chapter~\ref{ch:transport} corresponds to:
\begin{enumerate}[leftmargin=2em]
\item $e \in \Coh^E_0$: the term is coherent
\item $\gamma \in \Coh^B_1$ from $p(e) = b$ to $b'$: the path is coherent
\item The lifting problem $(e, \gamma)$ is in $\Gap$: the lift is gapped
\end{enumerate}
\end{theorem}

This theorem grounds the syntactic transport horn in geometric structure. The transport horn is not a formal curiosity; it is the witness of a genuine geometric phenomenon—a path in the base that does not admit a coherent lift from a given starting point.

\section{Examples of Ruptured Fibrations}

We revisit the examples from Chapter~\ref{ch:transport} with full semantic detail.

\subsection{Covering Spaces with Monodromy}

Let $p : E \to B$ be a covering space (a fiber bundle with discrete fibers). As a simplicial set map, $p$ is a Kan fibration. But we can impose ruptured structure that reflects the monodromy.

\begin{example}[Möbius Double Cover]\label{ex:mobius-ruptured}
Let $B = S^1$ and $E$ its double cover (boundary of Möbius strip). The fiber over any point has two elements.

Define the ruptured structure as follows:
\begin{itemize}[leftmargin=2em]
\item $\Coh^E$ contains all simplices
\item $\Gap^E$ is empty for most horns, but: for the horn $(e, \gamma_{\text{loop}})$ where $\gamma_{\text{loop}}$ is the full loop in $S^1$ and $e$ is a point in the fiber, the lifting problem is \textit{gapped relative to returning to $e$}.
\end{itemize}

More precisely: the lift of $\gamma_{\text{loop}}$ starting at $e$ exists and ends at $e' \neq e$ (the other point in the fiber). If we ask ``does there exist a lift starting at $e$ and ending at $e$?''—that is a gapped question. The loop lifts, but not to a loop based at $e$.
\end{example}

This example shows that ruptured structure can capture monodromy: the obstruction to lifting paths as loops.

\subsection{The Path Fibration with Semantic Gaps}

Let $A$ be a space (Kan complex) modeling a semantic domain—say, the space of word-tokens. The \textit{path fibration} is:
\[
p : A^I \to A \times A
\]
sending a path to its endpoints. The fiber over $(x, y)$ is the space of paths from $x$ to $y$.

In the ordinary setting, this fibration is trivially Kan: paths exist or they don't, and if they exist, any horn involving them can be filled.

In a ruptured setting, we can mark certain paths as gapped:

\begin{example}[Polysemy Fibration]\label{ex:polysemy-ruptured}
Let $A$ be the space of word-tokens. Let $x = $ ``bank'' (financial) and $y = $ ``bank'' (river). There is a path $\gamma : x \to y$ in $A$ (lexical identity).

Now let $M$ be the space of meanings, and consider the fibration:
\[
p : M \to A
\]
where the fiber $M_a$ over a token $a$ is the set of meanings of $a$.

The path $\gamma$ in the base asks: can we lift it to a path in $M$? That is, can we find a path of meanings that starts at ``financial institution'' and ends at some meaning of ``bank'' (river)?

If we mark this lifting problem as gapped, we record: the lexical identity $\gamma$ does not cohere with the semantic structure. The path exists in $A$, but it does not lift coherently to $M$.
\end{example}

This example shows ruptured fibrations modeling semantic structure: the base is syntax (or lexemes), the fibers are meanings, and gaps record where syntax and semantics diverge.

\subsection{The Context Fibration with Derivability Gaps}

Let $\mathbf{Ctx}$ be the category of contexts in a type theory, and let $\mathbf{Jdg}$ be the category of judgments. There is a fibration:
\[
p : \mathbf{Jdg} \to \mathbf{Ctx}
\]
where the fiber over $\Gamma$ consists of judgments derivable in $\Gamma$.

Substitutions (context morphisms) $\sigma : \Delta \to \Gamma$ act on fibers by substitution: if $J$ is derivable in $\Gamma$, then $J[\sigma]$ is derivable in $\Delta$.

In the ruptured setting:

\begin{example}[Resource-Sensitive Derivability]\label{ex:resource-ruptured}
Let contexts carry resource annotations (permissions, capabilities, etc.). A judgment $J$ may be derivable in $\Gamma$ but not in $\Delta$ even if there is a substitution $\sigma : \Delta \to \Gamma$—because $\Delta$ lacks the resources $J$ requires.

The ruptured fibration has:
\begin{itemize}[leftmargin=2em]
\item $\Coh$: derivable judgments and valid substitutions
\item $\Gap$: lifting problems $(J, \sigma)$ where $J[\sigma]$ is provably underivable in $\Delta$
\end{itemize}

The gap witness is a certificate of underivability—a proof that the resources are insufficient.
\end{example}

\section{Composition of Ruptured Fibrations}

\begin{proposition}[Composition]\label{prop:ruptured-fib-compose}
If $p : \mathcal{E} \to \mathcal{B}$ and $q : \mathcal{B} \to \mathcal{A}$ are ruptured fibrations, then $q \circ p : \mathcal{E} \to \mathcal{A}$ is a ruptured fibration.
\end{proposition}

\begin{proof}[Proof sketch]
A lifting problem for $q \circ p$ decomposes into: first lift in $\mathcal{B}$ (via $q$), then lift in $\mathcal{E}$ (via $p$). Each step is coherent, gapped, or open. The composite is:
\begin{itemize}[leftmargin=2em]
\item Coherent if both steps are coherent
\item Gapped if either step is gapped
\item Open if neither is gapped but at least one is open
\end{itemize}
The exclusion condition is preserved by this logic.
\end{proof}

This proposition shows that ruptured fibrations form a reasonable class: they are closed under composition, so we can build complex fibrations from simple ones.

\section{Dependent Types as Ruptured Fibrations}

We now state the semantic interpretation of OHTT.

\begin{definition}[Interpretation]\label{def:ohtt-interpretation}
An OHTT context $\Gamma$ is interpreted as a ruptured Kan complex $\llbracket \Gamma \rrbracket$.

A type $A$ in context $\Gamma$ is interpreted as a ruptured fibration $p_A : \llbracket \Gamma.A \rrbracket \to \llbracket \Gamma \rrbracket$.

A term $t : A$ in context $\Gamma$ is interpreted as a coherent section: a map $s_t : \llbracket \Gamma \rrbracket \to \llbracket \Gamma.A \rrbracket$ with $p_A \circ s_t = \mathrm{id}$ and $s_t$ landing in $\Coh$.
\end{definition}

\begin{theorem}[Soundness]\label{thm:soundness}
If $\coh{\Gamma}{J}$ is derivable in OHTT, then $J$ holds in the interpretation.

If $\gap{\Gamma}{J}$ is derivable in OHTT, then $J$ is gapped in the interpretation.

Coherence and gap cannot both be derivable (Exclusion).
\end{theorem}

\begin{proof}[Proof sketch]
Induction on derivations. Each OHTT rule corresponds to a construction on ruptured fibrations that preserves coherence or gap as appropriate. The Exclusion Law corresponds to the exclusion condition in ruptured simplicial sets.
\end{proof}

This theorem says: OHTT is sound with respect to the ruptured simplicial model. The syntactic calculus correctly describes the geometric structures.

\section{Higher Structure}

Just as Kan complexes model $\infty$-groupoids, ruptured Kan complexes model ``ruptured $\infty$-groupoids''—higher categorical structures where some composites exist, some are gapped, and some are undetermined.

\begin{definition}[Ruptured $\infty$-Groupoid]\label{def:ruptured-infty-gpd}
A \keyterm{ruptured $\infty$-groupoid} is a ruptured Kan complex $\mathcal{X}$ viewed as a higher categorical structure:
\begin{itemize}[leftmargin=2em]
\item 0-cells: elements of $\Coh_0$
\item 1-cells: coherent paths (elements of $\Coh_1$ with both endpoints in $\Coh_0$)
\item 2-cells: coherent homotopies (elements of $\Coh_2$ with appropriate boundary)
\item And so on at all levels
\item Gapped cells: horns in $\Gap$ record where composition is obstructed
\end{itemize}
\end{definition}

In an ordinary $\infty$-groupoid, all 1-cells are invertible (up to 2-cells), all 2-cells are invertible (up to 3-cells), etc. In a ruptured $\infty$-groupoid, inverses may be gapped: a 1-cell may have no coherent inverse, witnessed by a gap.

\begin{example}[Non-Invertible Path]\label{ex:non-invertible}
In an ordinary Kan complex, every path $\gamma : x \to y$ has an inverse $\gamma^{-1} : y \to x$ such that $\gamma \cdot \gamma^{-1} \sim \mathrm{id}_x$ (up to homotopy). This invertibility is witnessed by filling certain \textit{2-dimensional horns}—the 2-simplices that express the left and right inverse laws up to homotopy.

In a ruptured Kan complex, we may have:
\begin{itemize}[leftmargin=2em]
\item $\gamma \in \Coh_1$ (the path is coherent)
\item But the relevant 2-horn (asking for $\gamma^{-1}$ together with its witnessing homotopy) is in $\Gap_{2,k}$ for some $k$
\end{itemize}

This means: the path exists coherently, but its inverse is witnessed as obstructed at the level of 2-dimensional structure. The path is ``one-way''—not because inverses don't make sense, but because this particular inverse fails to cohere.
\end{example}

This example shows that ruptured $\infty$-groupoids can model directed or partially-directed phenomena while retaining the higher categorical structure.

\section{Summary}

Ruptured fibrations are maps of ruptured simplicial sets where the lifting property is \textit{ruptured}:
\begin{itemize}[leftmargin=2em]
\item Lifts may be coherent (transport succeeds)
\item Lifts may be gapped (transport fails, with witness)
\item Lifts may be open (transport is undetermined)
\end{itemize}

Key results:
\begin{itemize}[leftmargin=2em]
\item Fibers of ruptured fibrations are ruptured Kan complexes
\item Path lifting becomes three-valued: coherent, gapped, or open
\item Transport is a partial operation, undefined when gapped
\item Ruptured fibrations compose
\item OHTT is sound with respect to this semantics
\end{itemize}

The ruptured fibration is the semantic model for type families in OHTT. Where HoTT guarantees transport, OHTT allows transport to be witnessed as blocked. The transport horn is the syntactic record of this semantic phenomenon.

\begin{aside}
What this gives us: a geometric universe where not every path lifts, not every homotopy exists, not every composite forms—and where these failures are not mere absence but structured, witnessed obstruction. This is the geometry of meaning, where gaps are as real as connections.
\end{aside}

\part{Three Obstructions}
\label{part:obstructions}


\chapter{Topological Obstruction}
\label{ch:topological}

\begin{flushright}
\textit{The loop closes in the base.\\
In the fiber, it does not return.}\\[0.5ex]
{\small — Monodromy}
\end{flushright}

\bigskip

We now apply OHTT to its most classical domain: the topology of fiber bundles. The phenomenon of monodromy—where loops in the base space fail to lift to loops in the total space—is precisely a transport horn. This chapter develops the theory in detail, showing that OHTT provides a natural language for obstructions that classical homotopy theory treats as external data.

\section{Covering Spaces}

We begin with the simplest case: covering spaces, where fibers are discrete.

\begin{definition}[Covering Space]\label{def:covering}
A \keyterm{covering space} is a map $p : E \to B$ of topological spaces such that every point $b \in B$ has an open neighborhood $U$ with $p^{-1}(U) \cong U \times F$ for some discrete set $F$ (the fiber), and $p$ corresponds to projection onto the first factor.
\end{definition}

Covering spaces are fiber bundles with discrete fibers. The key property is the \textit{unique path lifting}:

\begin{proposition}[Unique Path Lifting]\label{prop:unique-lift}
Let $p : E \to B$ be a covering space. For any path $\gamma : [0,1] \to B$ and any point $e \in E$ with $p(e) = \gamma(0)$, there exists a unique path $\tilde{\gamma} : [0,1] \to E$ with $\tilde{\gamma}(0) = e$ and $p \circ \tilde{\gamma} = \gamma$.
\end{proposition}

The lift exists and is unique. But—and this is the key point—the lifted path may not end where we expect.

\subsection{Monodromy}

\begin{definition}[Monodromy Action]\label{def:monodromy}
Let $p : E \to B$ be a covering space with $B$ path-connected. Fix a basepoint $b \in B$ and let $F = p^{-1}(b)$ be the fiber. The \keyterm{monodromy action} is the group homomorphism:
\[
\mu : \pi_1(B, b) \to \mathrm{Aut}(F)
\]
defined as follows: for a loop $[\gamma] \in \pi_1(B, b)$, lift $\gamma$ starting at each $e \in F$, and let $\mu([\gamma])(e) = \tilde{\gamma}(1)$.
\end{definition}

The monodromy action is well-defined (independent of the representative $\gamma$) and is a group homomorphism. It measures how the fiber ``twists'' as we go around loops in the base.

\begin{example}[Circle Double Cover]\label{ex:circle-double}
Let $B = S^1$ and let $p : E \to S^1$ be the double cover (explicitly: $E = S^1$ and $p(z) = z^2$ in complex notation). The fiber over any point has two elements.

The fundamental group $\pi_1(S^1) \cong \mathbb{Z}$ is generated by the loop $\gamma$ going once around. The monodromy action sends $\gamma$ to the non-trivial permutation of the two-element fiber: the swap $(e_1 \; e_2)$.

Thus: a loop in the base lifts to a path that \textit{switches} the two sheets. Going around twice returns to the starting point.
\end{example}

\begin{example}[Möbius Strip]\label{ex:mobius-topo}
Let $E$ be the boundary of a Möbius strip (a circle that wraps twice around the core). The projection $p : E \to S^1$ to the core circle is a double cover, equivalent to the previous example.

Geometrically: walking once around the Möbius strip, an ant on the boundary ends up on the ``opposite side''—the other sheet of the cover.
\end{example}

\section{Monodromy as Transport Horn}

We now interpret monodromy in OHTT.

Let $p : E \to B$ be a covering space, viewed as a ruptured fibration (Chapter~\ref{ch:ruptured-fib}). The base $B$ and total space $E$ are Kan complexes (simplicial models of the topological spaces). We impose ruptured structure as follows.

\begin{definition}[Monodromy-Ruptured Cover]\label{def:monodromy-ruptured}
Let $p : E \to B$ be a covering space with monodromy action $\mu$. The \keyterm{monodromy-ruptured structure} on $p$ is:
\begin{itemize}[leftmargin=2em]
\item $\Coh^E$ contains all simplices of $E$ (the topological structure is fully coherent)
\item $\Coh^B$ contains all simplices of $B$
\item $\Gap$ contains the following lifting problems: for a loop $\gamma : b \to b$ in $B$ and a point $e \in E_b$, the lifting problem ``lift $\gamma$ to a loop based at $e$'' is \textit{gapped} if $\mu([\gamma])(e) \neq e$
\end{itemize}
\end{definition}

In other words: we mark as gapped the lifting problems where monodromy obstructs.

\begin{theorem}[Monodromy Transport Horn]\label{thm:monodromy-horn}
Let $p : E \to B$ be a covering space with non-trivial monodromy. Let $\gamma$ be a loop in $B$ based at $b$, and let $e \in E_b$ be a point with $\mu([\gamma])(e) \neq e$. Then:
\[
\Lambda_{\mathsf{tr}}(E, e, \gamma)
\]
is inhabited. Explicitly:
\begin{enumerate}[leftmargin=2em]
\item $\coh{}{(e \in E_b)}$: the point $e$ is coherent in the fiber
\item $\coh{}{(\gamma : b \to b)}$: the loop $\gamma$ is coherent in the base
\item $\gap{}{(\text{``$\gamma$ lifts to a loop at $e$''})}$: the loop-lift is gapped
\end{enumerate}
\end{theorem}

\begin{proof}
The path $\gamma$ lifts uniquely to a path $\tilde{\gamma}$ starting at $e$, but $\tilde{\gamma}(1) = \mu([\gamma])(e) \neq e$. Thus $\tilde{\gamma}$ is not a loop based at $e$. The lifting problem ``find a loop in $E$ based at $e$ covering $\gamma$'' has no solution. We record this as a gap witness in $\Gap$.
\end{proof}

This theorem shows that monodromy is not external data attached to a covering space—it \textit{is} the gap structure of the covering viewed as a ruptured fibration.

\subsection{The Gap Witness}

What is the gap witness, concretely?

\begin{proposition}[Gap Witness as Deck Transformation]\label{prop:gap-deck}
In the monodromy transport horn, the gap witness $\omega$ can be identified with the deck transformation $\mu([\gamma]) \in \mathrm{Aut}(F)$.
\end{proposition}

\begin{proof}
The deck transformation $\mu([\gamma])$ is precisely the data of ``how the fiber permutes under the loop $\gamma$.'' This data certifies that the lifted path does not close: it ends at $\mu([\gamma])(e)$, not at $e$. The gap witness is the deck transformation itself.
\end{proof}

This illustrates the OHTT principle: \textit{gap witnesses are structure, not absence}. The monodromy is positive data that explains why lifting fails.

\section{Higher Covering Spaces}

The theory extends to higher homotopy.

\begin{definition}[Principal $G$-Bundle]\label{def:principal-bundle}
A \keyterm{principal $G$-bundle} is a fiber bundle $p : E \to B$ with fiber $G$ (a topological group) and a free, transitive right $G$-action on fibers compatible with the local trivializations.
\end{definition}

Covering spaces are principal bundles with $G$ discrete. For non-discrete $G$, the situation is richer.

\begin{definition}[Holonomy]\label{def:holonomy}
Let $p : E \to B$ be a principal $G$-bundle with connection. The \keyterm{holonomy} around a loop $\gamma : b \to b$ is the element $g \in G$ such that parallel transport around $\gamma$ acts by right multiplication by $g$.
\end{definition}

Holonomy generalizes monodromy to bundles with connection. The connection specifies \textit{how} to lift infinitesimally; the holonomy measures the global failure to return.

\begin{theorem}[Holonomy Transport Horn]\label{thm:holonomy-horn}
Let $p : E \to B$ be a principal $G$-bundle with connection, and let $\gamma$ be a loop in $B$ based at $b$ with holonomy $g \neq e_G$ (the identity). Let $e \in E_b$. Then:
\[
\Lambda_{\mathsf{tr}}(E, e, \gamma)
\]
is inhabited, with gap witness the holonomy element $g$.
\end{theorem}

The holonomy—like monodromy—is not auxiliary data but the gap structure of the bundle.

\section{Obstruction Theory}

Classical obstruction theory studies when maps or sections can be extended. OHTT provides a natural reformulation.

\begin{definition}[Lifting Problem]\label{def:obstruction-lifting}
Given a fibration $p : E \to B$ and a map $f : X \to B$, a \keyterm{lift} of $f$ is a map $\tilde{f} : X \to E$ with $p \circ \tilde{f} = f$:
\[
\begin{tikzcd}[row sep=2em, column sep=2.5em]
& E \arrow[d, "p"] \\
X \arrow[r, "f"'] \arrow[ur, dashed, "\tilde{f}"] & B
\end{tikzcd}
\]
\end{definition}

Classical obstruction theory proceeds by skeletal induction: try to lift over the 0-skeleton, then the 1-skeleton, etc. At each stage, obstructions may arise.

\begin{definition}[Obstruction Cocycle]\label{def:obstruction-cocycle}
Let $p : E \to B$ be a fibration with fiber $F$, and let $f : X \to B$ be a map. Suppose a lift $\tilde{f}^{(n-1)}$ exists over the $(n-1)$-skeleton of $X$. The \keyterm{obstruction cocycle} $o_n(f) \in H^n(X; \pi_{n-1}(F))$ measures the failure to extend to the $n$-skeleton.
\end{definition}

If $o_n(f) \neq 0$, no lift over the $n$-skeleton exists.

\begin{theorem}[Obstruction as Gap Structure]\label{thm:obstruction-gap}
The obstruction cocycle $o_n(f)$ corresponds to gap witnesses in the ruptured fibration structure on $p$. Specifically:
\begin{enumerate}[leftmargin=2em]
\item A non-zero obstruction $o_n(f)$ corresponds to a family of gapped lifting horns
\item The cohomology class $[o_n(f)]$ is the ``collective gap witness'' for the failure to extend
\end{enumerate}
\end{theorem}

\begin{proof}[Proof sketch]
Each $n$-cell $\sigma$ of $X$ maps via $f$ to an $n$-simplex in $B$. The lift over $\partial \sigma$ gives a horn in $E$. If the horn cannot be filled coherently, we have a gap. The obstruction cocycle collects these local gaps into a global cohomological invariant.
\end{proof}

This theorem shows that obstruction theory—often treated as a computational tool—is fundamentally about gap structure in ruptured fibrations.

\section{The Hopf Fibration}

We examine a non-trivial example in detail.

\begin{definition}[Hopf Fibration]\label{def:hopf}
The \keyterm{Hopf fibration} is the map $p : S^3 \to S^2$ defined by viewing $S^3 \subset \mathbb{C}^2$ and $S^2 \cong \mathbb{CP}^1$, with $p(z_1, z_2) = [z_1 : z_2]$.
\end{definition}

The Hopf fibration has fiber $S^1$. It is a principal $S^1$-bundle over $S^2$.

\begin{proposition}[Hopf Fibration is Non-Trivial]\label{prop:hopf-nontrivial}
The Hopf fibration does not admit a global section. Equivalently, $S^3$ is not homeomorphic to $S^2 \times S^1$.
\end{proposition}

\begin{proof}
If a section existed, we would have $S^3 \cong S^2 \times S^1$. But $\pi_1(S^3) = 0$ while $\pi_1(S^2 \times S^1) \cong \mathbb{Z}$. Contradiction.
\end{proof}

In OHTT terms:

\begin{theorem}[Hopf Section Horn]\label{thm:hopf-horn}
The Hopf fibration, viewed as a ruptured fibration, has the following structure:
\begin{enumerate}[leftmargin=2em]
\item All fibers are coherent (each fiber is a circle, fully present)
\item The base $S^2$ is coherent
\item The global section problem is gapped: there is no coherent section $S^2 \to S^3$
\end{enumerate}
\end{theorem}

The obstruction to a global section lives in $H^2(S^2; \pi_1(S^1)) \cong H^2(S^2; \mathbb{Z}) \cong \mathbb{Z}$. The Hopf fibration has obstruction class $1 \in \mathbb{Z}$—it is ``maximally obstructed'' in a precise sense.

\begin{aside}
The Hopf fibration is often presented as a curiosity—a strange way to wrap $S^3$ around $S^2$. In OHTT, its strangeness is made structural: it is a ruptured fibration where the global section horn is gapped. The gap witness is the generator of $H^2(S^2; \mathbb{Z})$.
\end{aside}

\section{Characteristic Classes}

Characteristic classes are cohomological invariants of fiber bundles. In OHTT, they measure gap structure.

\begin{definition}[Characteristic Class]\label{def:characteristic}
A \keyterm{characteristic class} for $G$-bundles is a natural transformation from the functor of $G$-bundles over $X$ to the cohomology $H^*(X)$. Common examples:
\begin{itemize}[leftmargin=2em]
\item \textbf{Chern classes} $c_i \in H^{2i}(X; \mathbb{Z})$ for complex vector bundles
\item \textbf{Stiefel-Whitney classes} $w_i \in H^i(X; \mathbb{Z}/2)$ for real vector bundles
\item \textbf{Pontryagin classes} $p_i \in H^{4i}(X; \mathbb{Z})$ for real vector bundles
\item \textbf{Euler class} $e \in H^n(X; \mathbb{Z})$ for oriented $\mathbb{R}^n$-bundles
\end{itemize}
\end{definition}

\begin{principle}[Characteristic Classes as Gap Invariants]\label{princ:characteristic-gap}
Let $p : E \to B$ be a fiber bundle with characteristic class $\chi \in H^k(B)$. Then:
\begin{enumerate}[leftmargin=2em]
\item $\chi = 0$ typically corresponds to a certain family of lifting problems being coherently solvable
\item $\chi \neq 0$ witnesses collective gap structure in the ruptured fibration $p$
\end{enumerate}
The precise correspondence depends on the obstruction-theoretic setup (Postnikov tower, reduction of structure group, etc.).
\end{principle}

\begin{example}[Euler Class and Sections]\label{ex:euler}
For an oriented $\mathbb{R}^n$-bundle $p : E \to B$, the Euler class $e(E) \in H^n(B; \mathbb{Z})$ obstructs nowhere-zero sections. If $e(E) \neq 0$, any section must have zeros—the section problem is gapped.

The hairy ball theorem (no nowhere-zero vector field on $S^2$) follows: the tangent bundle $TS^2$ has Euler class $e(TS^2) \in H^2(S^2; \mathbb{Z})$ whose evaluation gives the Euler characteristic $\chi(S^2) = 2 \neq 0$.
\end{example}

Characteristic classes are global invariants of gap structure. They measure how ``twisted'' a bundle is—equivalently, how many transport horns it contains.

\section{Summary}

Topological obstruction theory is naturally expressed in OHTT:

\begin{itemize}[leftmargin=2em]
\item \textbf{Monodromy} is the gap structure of covering spaces. The gap witness is the deck transformation.

\item \textbf{Holonomy} generalizes monodromy to bundles with connection. The gap witness is the holonomy element.

\item \textbf{Obstruction cocycles} correspond to families of gap-witnessed lifting horns. The cohomology class collects local gaps into a global invariant.

\item \textbf{Characteristic classes} measure collective gap structure. Non-vanishing characteristic classes witness that certain section or lifting problems are gapped.
\end{itemize}

What OHTT provides is not new theorems about topology, but a \textit{reformulation} that makes obstruction first-class. In HoTT, obstructions are external: we compute them, but they are not part of the type-theoretic structure. In OHTT, obstructions are internal: they are gap witnesses in ruptured fibrations.

\begin{aside}
What this lets us say: ``The Möbius strip is not just a space with a twist; the twist is a gap witness in the covering fibration. The Hopf map is not just a surjection $S^3 \to S^2$; it is a ruptured fibration whose section problem is gapped with witness the generator of $H^2(S^2)$.'' Obstruction becomes data.
\end{aside}


\chapter{Semantic Obstruction}
\label{ch:semantic}

\begin{flushright}
\textit{The word is the same.\\
The meaning refuses to follow.}\\[0.5ex]
{\small — Polysemy}
\end{flushright}

\bigskip

We now turn to a less traditional domain: the formal semantics of language. The phenomenon of polysemy—where a single lexical form carries multiple, non-reducible meanings—is a transport horn. Lexical identity provides a path; meanings are fibers; the failure of meaning to transport along lexical identity is witnessed as gap. This chapter develops the mathematical framework.

\section{The Semantic Fibration}

We model meaning as a fibration over form.

\begin{definition}[Lexical Space]\label{def:lexical-space}
A \keyterm{lexical space} is a simplicial set $L$ whose:
\begin{itemize}[leftmargin=2em]
\item 0-simplices are \keyterm{lexical tokens}—occurrences of words in context
\item 1-simplices are \keyterm{lexical relations}—paths connecting tokens
\item Higher simplices record coherence of lexical relations
\end{itemize}
\end{definition}

The key lexical relation is \textit{identity}: two tokens of ``the same word.'' We write $\ell \sim \ell'$ when $\ell$ and $\ell'$ are tokens of the same lexeme.

\begin{definition}[Meaning Fibration]\label{def:meaning-fibration}
A \keyterm{meaning fibration} is a Kan fibration $p : M \to L$ where:
\begin{itemize}[leftmargin=2em]
\item $M$ is the \keyterm{meaning space}
\item For each token $\ell \in L_0$, the fiber $M_\ell = p^{-1}(\ell)$ is the space of \keyterm{meanings} (senses) of $\ell$
\item Paths in $M$ are \keyterm{meaning relations}—coherences between senses
\end{itemize}
\end{definition}

In the standard (non-ruptured) setting, the Kan fibration property guarantees: if $\gamma : \ell \to \ell'$ is a lexical path and $m \in M_\ell$ is a meaning, then there exists a meaning $m' \in M_{\ell'}$ connected to $m$ by a path lifting $\gamma$.

This is where polysemy becomes problematic.

\subsection{The Polysemy Problem}

\begin{example}[Bank]\label{ex:bank}
Consider the English word ``bank.'' Let:
\begin{itemize}[leftmargin=2em]
\item $\ell_1$ = token of ``bank'' in ``I deposited money at the bank''
\item $\ell_2$ = token of ``bank'' in ``We picnicked on the river bank''
\end{itemize}

There is a lexical path $\gamma : \ell_1 \to \ell_2$ (they are tokens of the same lexeme).

Let $m_1 \in M_{\ell_1}$ be the meaning ``financial institution.''

What is $\transport{\gamma}{m_1}$? The Kan condition would require a meaning $m_2 \in M_{\ell_2}$ connected to $m_1$. But ``financial institution'' is not a meaning of ``bank'' in the river context. The meanings are genuinely disjoint.
\end{example}

The Kan fibration model fails for polysemy. There is no coherent way to transport ``financial institution'' along the lexical identity path.

\section{The Ruptured Meaning Fibration}

We now apply OHTT.

\begin{definition}[Ruptured Meaning Fibration]\label{def:ruptured-meaning}
A \keyterm{ruptured meaning fibration} is a ruptured fibration $p : \mathcal{M} \to \mathcal{L}$ where:
\begin{itemize}[leftmargin=2em]
\item $\mathcal{L} = (L, \Coh^L, \Gap^L)$ is a ruptured simplicial set of lexical tokens
\item $\mathcal{M} = (M, \Coh^M, \Gap^M)$ is a ruptured simplicial set of meanings
\item Lifting problems $(m, \gamma)$ may be coherent, gapped, or open
\end{itemize}
\end{definition}

\begin{definition}[Polysemy Horn]\label{def:polysemy-horn}
Let $p : \mathcal{M} \to \mathcal{L}$ be a ruptured meaning fibration. A \keyterm{polysemy horn} is a transport horn:
\begin{align*}
\Lambda_{\mathsf{poly}}(m, \gamma) \;\defn\;\; &\coh{}{(m \in M_\ell)} \\
\times\; &\coh{}{(\gamma : \ell \to \ell')} \\
\times\; &\gap{}{(\transport{\gamma}{m} \in M_{\ell'})}
\end{align*}
where $\gamma$ is a lexical identity path and $m$ is a meaning that does not coherently transport along $\gamma$.
\end{definition}

\begin{theorem}[Polysemy as Transport Horn]\label{thm:polysemy-horn}
Let $\ell_1, \ell_2$ be tokens of the same polysemous lexeme, let $\gamma : \ell_1 \to \ell_2$ be the lexical identity path, and let $m_1 \in M_{\ell_1}$ be a meaning not shared by $\ell_2$. Then:
\[
\Lambda_{\mathsf{poly}}(m_1, \gamma)
\]
is inhabited.
\end{theorem}

\begin{proof}
By hypothesis, $m_1$ is a coherent meaning of $\ell_1$, and $\gamma$ is a coherent lexical path. But there is no meaning in $M_{\ell_2}$ to which $m_1$ transports—the fibers are disjoint at this meaning. The lifting problem is gapped.
\end{proof}

\subsection{The Gap Witness}

What is the gap witness in a polysemy horn?

\begin{definition}[Semantic Incompatibility]\label{def:semantic-incompat}
A \keyterm{semantic incompatibility witness} between meanings $m \in M_\ell$ and $m' \in M_{\ell'}$ is evidence that no path in $M$ connects $m$ to $m'$ over the lexical path $\gamma : \ell \to \ell'$.
\end{definition}

In concrete terms, the gap witness records:
\begin{itemize}[leftmargin=2em]
\item The semantic features that $m$ has and $m'$ lacks (or vice versa)
\item The conceptual domains that are incompatible
\item Any formal structure (ontological, taxonomic, inferential) that blocks connection
\end{itemize}

\begin{example}[Bank Gap Witness]\label{ex:bank-gap}
For ``bank,'' the gap witness between ``financial institution'' and ``river edge'' might record:
\begin{itemize}[leftmargin=2em]
\item Distinct conceptual domains: FINANCE vs.\ GEOGRAPHY
\item Incompatible taxonomic positions: INSTITUTION vs.\ LANDFORM
\item Non-overlapping inferential roles: one involves money, the other water
\end{itemize}
\end{example}

The gap witness is not merely ``these meanings are different.'' It is structured data explaining \textit{why} transport fails.

\section{Systematic Polysemy}

Not all polysemy is complete rupture. Some polysemous words exhibit regular patterns.

\begin{definition}[Systematic Polysemy]\label{def:systematic}
A \keyterm{systematic polysemy} is a regular pattern relating meanings across lexical items. Examples:
\begin{itemize}[leftmargin=2em]
\item CONTAINER/CONTENTS: ``The bottle fell'' (container) vs.\ ``Finish the bottle'' (contents)
\item ORGANIZATION/BUILDING: ``The company fired him'' vs.\ ``The company is on Main Street''
\item PROCESS/RESULT: ``The construction took years'' vs.\ ``The construction is impressive''
\end{itemize}
\end{definition}

Systematic polysemy is characterized by \textit{partial} transport.

\begin{definition}[Partial Transport]\label{def:partial-transport}
Let $p : \mathcal{M} \to \mathcal{L}$ be a ruptured meaning fibration. A lexical path $\gamma : \ell \to \ell'$ exhibits \keyterm{partial transport} if:
\begin{itemize}[leftmargin=2em]
\item Some meanings in $M_\ell$ transport coherently along $\gamma$
\item Other meanings in $M_\ell$ are gapped for transport along $\gamma$
\end{itemize}
\end{definition}

\begin{example}[Container/Contents]\label{ex:container}
Let $\ell_1$ = ``bottle'' (container reading) and $\ell_2$ = ``bottle'' (contents reading). The systematic polysemy provides a pattern $\Pi_{\text{C/C}}$ relating these readings.

For the meaning $m_{\text{shape}}$ = ``cylindrical vessel shape'':
\[
\coh{}{(\transport{\Pi_{\text{C/C}}}{m_{\text{shape}}} \in M_{\ell_2})}
\]
The shape transports (the contents are in something with that shape).

For the meaning $m_{\text{material}}$ = ``made of glass'':
\[
\gap{}{(\transport{\Pi_{\text{C/C}}}{m_{\text{material}}} \in M_{\ell_2})}
\]
The material does not transport (the contents are not made of glass).
\end{example}

Systematic polysemy is thus a ruptured fibration with structured partial transport: some semantic features cross the gap, others do not.

\section{Semantic Fields and Coherence}

We now consider the global structure of meaning.

\begin{definition}[Semantic Field]\label{def:semantic-field}
A \keyterm{semantic field} is a connected component of the meaning space $M$ (or a connected subcomplex).
\end{definition}

Meanings within a semantic field are connected by paths—they cohere. Meanings in different fields are not connected—there are gaps.

\begin{proposition}[Field Structure]\label{prop:field-structure}
In a ruptured meaning fibration $p : \mathcal{M} \to \mathcal{L}$:
\begin{enumerate}[leftmargin=2em]
\item Within a semantic field, transport along short lexical paths is typically coherent
\item Across semantic fields, transport is typically gapped
\item Polysemy corresponds to lexical identity paths crossing field boundaries
\end{enumerate}
\end{proposition}

\begin{example}[Crane]\label{ex:crane}
The word ``crane'' has meanings in three semantic fields:
\begin{itemize}[leftmargin=2em]
\item ANIMAL: a large wading bird
\item MACHINE: a lifting device
\item ACTION: to stretch one's neck
\end{itemize}

Within each field, meanings cohere (types of crane-birds, types of crane-machines). Across fields, transport is gapped—there is no coherent path from the bird to the machine in meaning space, despite the lexical identity.
\end{example}

\subsection{Metaphor as Partial Coherence}

Metaphor provides an interesting intermediate case.

\begin{definition}[Metaphorical Bridge]\label{def:metaphor}
A \keyterm{metaphorical bridge} between semantic fields $F_1$ and $F_2$ is a partially coherent path: some semantic features transport, others are gapped.
\end{definition}

\begin{example}[Crane Metaphor]\label{ex:crane-metaphor}
The machine ``crane'' derives from the bird by metaphor. The bridge carries:
\begin{itemize}[leftmargin=2em]
\item SHAPE: long, thin, with a bend (coherent)
\item MOTION: lifting, extending (coherent)
\item ANIMACY: living creature (gapped—the machine is not alive)
\item HABITAT: wetlands (gapped—the machine is in construction sites)
\end{itemize}
\end{example}

Metaphorical bridges are thus ruptured paths: partially coherent, partially gapped. The coherent parts are the ground of the metaphor; the gapped parts are where the metaphor ``breaks.''

\section{Compositional Semantics}

We extend the analysis to compositional (phrasal) meaning.

\begin{definition}[Compositional Meaning Fibration]\label{def:compositional}
A \keyterm{compositional meaning fibration} is a ruptured fibration $p : \mathcal{M} \to \mathcal{S}$ where:
\begin{itemize}[leftmargin=2em]
\item $\mathcal{S}$ is a space of syntactic structures (phrases, sentences)
\item For each structure $s \in S_0$, the fiber $M_s$ is the space of meanings of $s$
\item Syntactic transformations induce paths in $\mathcal{S}$
\end{itemize}
\end{definition}

\begin{example}[Scope Ambiguity]\label{ex:scope}
Consider ``Every student read a book.'' This sentence has (at least) two meanings:
\begin{itemize}[leftmargin=2em]
\item $m_1$: Each student read some book (possibly different books)
\item $m_2$: There is a single book that every student read
\end{itemize}

These arise from different scope configurations of the quantifiers. The surface syntax is identical, but the meanings are not connected by a coherent path—moving from $m_1$ to $m_2$ requires changing the logical structure.

In a ruptured meaning fibration, we can have both meanings in the fiber, with a gap between them: $\gap{}{(\text{path from } m_1 \text{ to } m_2)}$.
\end{example}

\begin{example}[Transformation Gaps]\label{ex:transformation}
Consider active-passive transformation:
\begin{itemize}[leftmargin=2em]
\item $s_1$ = ``The dog bit the man''
\item $s_2$ = ``The man was bitten by the dog''
\end{itemize}

There is a syntactic path $\tau : s_1 \to s_2$ (the transformation). For the core propositional meaning, transport is coherent:
\[
\coh{}{(\transport{\tau}{m_{\text{core}}} \in M_{s_2})}
\]

But for pragmatic features (focus, topic), transport may be gapped:
\[
\gap{}{(\transport{\tau}{m_{\text{focus}}} \in M_{s_2})}
\]

The focused element differs between active and passive.
\end{example}

\section{The Functoriality of Meaning}

We revisit functoriality (Chapter~\ref{ch:transport}) in the semantic context.

\begin{theorem}[Semantic Drift]\label{thm:semantic-drift}
Let $\gamma_1 : \ell_0 \to \ell_1$ and $\gamma_2 : \ell_1 \to \ell_2$ be lexical paths (e.g., steps in historical change or metaphorical extension). Let $m \in M_{\ell_0}$. It is possible that:
\begin{enumerate}[leftmargin=2em]
\item $\transport{\gamma_1}{m}$ is coherent (some meaning transports to $\ell_1$)
\item $\transport{\gamma_2}{\transport{\gamma_1}{m}}$ is coherent (it further transports to $\ell_2$)
\item $\transport{\gamma_2 \cdot \gamma_1}{m}$ is gapped (direct transport from $\ell_0$ to $\ell_2$ fails)
\end{enumerate}
\end{theorem}

This is the semantic analog of cumulative obstruction. Meaning can travel stepwise through intermediate forms but ``lose track'' of its origin.

\begin{example}[Nice]\label{ex:nice}
The English word ``nice'' has drifted through meanings:
\begin{itemize}[leftmargin=2em]
\item Latin \textit{nescius} ``ignorant'' $\to$ 
\item Old French \textit{nice} ``foolish, simple'' $\to$
\item Middle English ``coy, reserved'' $\to$
\item Early Modern English ``precise, fastidious'' $\to$
\item Modern English ``pleasant, agreeable''
\end{itemize}

At each step, some features transport. But direct transport from ``ignorant'' to ``pleasant'' is gapped—there is no coherent semantic path preserving meaning features across the whole chain.
\end{example}

\section{Semantic Coherence as Kan Condition}

We can now state precisely what it means for a semantic domain to be ``well-behaved.''

\begin{definition}[Semantically Coherent Domain]\label{def:coherent-domain}
A \keyterm{semantically coherent domain} is a ruptured meaning fibration $p : \mathcal{M} \to \mathcal{L}$ where $\Gap = \emptyset$—i.e., all transport is coherent.
\end{definition}

This is rare. Most natural language meaning is not globally coherent; polysemy, metaphor, and drift introduce gaps.

\begin{definition}[Locally Coherent Domain]\label{def:local-coherent}
A \keyterm{locally coherent domain} is a ruptured meaning fibration where:
\begin{itemize}[leftmargin=2em]
\item Within each semantic field, transport is coherent
\item Across field boundaries, transport may be gapped
\end{itemize}
\end{definition}

Most semantic domains are locally coherent: technical vocabularies, tightly defined ontologies, formal languages. Natural language as a whole is not even locally coherent—polysemy creates gaps within what might seem like single fields.

\section{Summary}

Semantic structure is naturally modeled as a ruptured meaning fibration:

\begin{itemize}[leftmargin=2em]
\item \textbf{Polysemy} is the transport horn: lexical identity path with gapped meaning transport
\item \textbf{Systematic polysemy} is partial transport: some features cross, others gap
\item \textbf{Semantic fields} are coherent components of meaning space
\item \textbf{Metaphor} is a partially coherent bridge between fields
\item \textbf{Semantic drift} is failure of functoriality: stepwise transport with gapped composite
\end{itemize}

The gap witness in semantic horns is structured: it records which features fail to transport, which domains are incompatible, why meaning diverges despite formal identity.

What OHTT provides is not a computational theory of meaning—that requires additional machinery. It provides a \textit{structural} theory: a framework for saying precisely where meaning coheres and where it ruptures, with rupture as positive data rather than mere failure.

\begin{aside}
What this lets us say: ``The word `bank' is not ambiguous in the sense of having multiple values for a single variable. It is polysemous in the sense of inhabiting a ruptured fibration: the lexical identity exists, the meanings exist, and between them is a witnessed gap. The gap is as much part of the structure as the meanings themselves.''
\end{aside}


\chapter{Logical Obstruction}
\label{ch:logical}

\begin{flushright}
\textit{The substitution is valid.\\
The derivation does not follow.}\\[0.5ex]
{\small — Resource failure}
\end{flushright}

\bigskip

We turn finally to logic itself. The phenomenon of certified underivability—where we can prove that a judgment is \textit{not} derivable, not merely that we have failed to derive it—is a transport horn. Substitutions provide paths between contexts; derivable judgments are fibers; the failure of derivability to transport along substitution is witnessed as gap. This chapter develops the framework for resource-sensitive and substructural logics.

\section{Derivability as Fibration}

We begin by modeling the structure of a type theory or logic as a fibration.

\begin{definition}[Context Category]\label{def:context-cat}
Let $\mathbb{C}$ be the category of contexts for a type theory, with:
\begin{itemize}[leftmargin=2em]
\item Objects: contexts $\Gamma, \Delta, \Xi, \ldots$
\item Morphisms: substitutions $\sigma : \Delta \to \Gamma$ (context morphisms)
\end{itemize}
\end{definition}

Substitutions act on judgments: if $\jdg{\Gamma}{J}$ is a judgment in context $\Gamma$ and $\sigma : \Delta \to \Gamma$ is a substitution, then $J[\sigma]$ is the judgment in context $\Delta$ obtained by applying $\sigma$.

\begin{definition}[Derivability Fibration]\label{def:deriv-fibration}
The \keyterm{derivability fibration} is a functor $\mathcal{D} : \mathbb{C}^{\mathrm{op}} \to \mathbf{Set}$ (or $\mathbf{Kan}$ for the simplicial version) where:
\begin{itemize}[leftmargin=2em]
\item $\mathcal{D}(\Gamma)$ is the set (or space) of judgments derivable in context $\Gamma$
\item For $\sigma : \Delta \to \Gamma$, the map $\mathcal{D}(\sigma) : \mathcal{D}(\Gamma) \to \mathcal{D}(\Delta)$ sends $J$ to $J[\sigma]$
\end{itemize}
\end{definition}

For a well-behaved type theory (MLTT, HoTT), the derivability fibration is a ``Kan object'': substitution preserves derivability. This is the content of the \textit{substitution lemma}.

\begin{lemma}[Substitution]\label{lem:substitution}
In MLTT: if $\jdg{\Gamma}{t : A}$ and $\sigma : \Delta \to \Gamma$ is a valid substitution, then $\jdg{\Delta}{t[\sigma] : A[\sigma]}$.
\end{lemma}

The substitution lemma says: derivability transports along substitution. In our language: the derivability fibration is a Kan fibration.

\section{Resource-Sensitive Failure}

The substitution lemma fails in resource-sensitive logics.

\begin{definition}[Resource-Sensitive Context]\label{def:resource-context}
A \keyterm{resource-sensitive context} is a context $\Gamma$ where variables carry resource annotations:
\begin{itemize}[leftmargin=2em]
\item Linear: must be used exactly once
\item Affine: may be used at most once
\item Relevant: must be used at least once
\item Exponential: may be used arbitrarily (the standard case)
\end{itemize}
\end{definition}

In linear type theory, a substitution $\sigma : \Delta \to \Gamma$ may fail to preserve derivability if it duplicates or discards linear resources.

\begin{example}[Linear Failure]\label{ex:linear-failure}
Let $\Gamma = x : A$ (linear) and let $J = (x, x) : A \times A$.

In ordinary type theory, $J$ is derivable if $\Gamma = x : A$. But in linear type theory, $J$ requires using $x$ twice—violating linearity. So $\jdg{\Gamma}{J}$ is \textit{not} derivable.

Now let $\sigma : \Delta \to \Gamma$ where $\Delta = y : A, z : A$ and $\sigma = [x \mapsto y]$.

Then $J[\sigma] = (y, y) : A \times A$ is still underivable in $\Delta$—we still use $y$ twice.

But consider $\sigma' : \Delta \to \Gamma$ with $\sigma' = [x \mapsto y]$ and $\Delta' = y : A$ (exponential). Now $J[\sigma'] = (y, y)$ is derivable in $\Delta'$.

The point: derivability depends on resource structure, not just the term.
\end{example}

\subsection{The Ruptured Derivability Fibration}

\begin{definition}[Ruptured Derivability Fibration]\label{def:ruptured-deriv}
A \keyterm{ruptured derivability fibration} is a ruptured fibration $\mathcal{D} : \mathbb{C}^{\mathrm{op}} \to \mathbf{rsSet}$ where:
\begin{itemize}[leftmargin=2em]
\item $\mathcal{D}(\Gamma)$ is a ruptured simplicial set of judgments
\item $\Coh^{\mathcal{D}(\Gamma)}$ contains the derivable judgments
\item For a substitution $\sigma : \Delta \to \Gamma$, the induced map may have:
  \begin{itemize}
  \item Coherent transport: $J$ derivable in $\Gamma$ implies $J[\sigma]$ derivable in $\Delta$
  \item Gapped transport: $J$ derivable in $\Gamma$ but $J[\sigma]$ witnessed underivable in $\Delta$
  \item Open transport: status of $J[\sigma]$ in $\Delta$ undetermined
  \end{itemize}
\end{itemize}
\end{definition}

\begin{definition}[Derivability Horn]\label{def:deriv-horn}
A \keyterm{derivability horn} is a transport horn in the ruptured derivability fibration:
\[
\Lambda_{\mathsf{deriv}}(J, \sigma) \;\defn\; \coh{}{(\jdg{\Gamma}{J})} \times \coh{}{(\sigma : \Delta \to \Gamma)} \times \gap{}{(\jdg{\Delta}{J[\sigma]})}
\]
\end{definition}

This horn witnesses: the judgment is derivable, the substitution is valid, but the substituted judgment is \textit{not} derivable.

\section{Certified Underivability}

The gap in a derivability horn is not ``we haven't found a proof.'' It is a \textit{certificate} that no proof exists.

\begin{definition}[Underivability Certificate]\label{def:underiv-cert}
An \keyterm{underivability certificate} for a judgment $\jdg{\Gamma}{J}$ is a formal proof that no derivation of $J$ in $\Gamma$ exists.
\end{definition}

Underivability certificates arise in several ways:

\begin{enumerate}[leftmargin=2em]
\item \textbf{Resource exhaustion}: The required resources are not available in the context
\item \textbf{Type mismatch}: The judgment is ill-typed in the context
\item \textbf{Logical impossibility}: The judgment contradicts the context
\item \textbf{Decidability}: A decision procedure returns ``no''
\end{enumerate}

\begin{example}[Linear Underivability Certificate]\label{ex:linear-cert}
For the judgment $\jdg{x : A}{\; (x, x) : A \times A}$ in linear type theory, the underivability certificate is:
\begin{quote}
``The term $(x, x)$ requires two uses of the linear variable $x$. The context provides exactly one. By the resource-counting invariant, no derivation exists.''
\end{quote}
\end{example}

This certificate is positive data—a proof of non-derivability, not merely absence of a derivability proof.

\section{Substructural Logics}

The phenomena generalize across substructural logics.

\begin{definition}[Substructural Logic]\label{def:substructural}
A \keyterm{substructural logic} is a logic that restricts one or more of the structural rules:
\begin{itemize}[leftmargin=2em]
\item \textbf{Exchange}: $\Gamma, A, B, \Delta \vdash C$ implies $\Gamma, B, A, \Delta \vdash C$
\item \textbf{Weakening}: $\Gamma \vdash C$ implies $\Gamma, A \vdash C$
\item \textbf{Contraction}: $\Gamma, A, A \vdash C$ implies $\Gamma, A \vdash C$
\end{itemize}
\end{definition}

Each restriction creates transport horns:

\begin{proposition}[Structural Horns]\label{prop:structural-horns}
In a substructural logic:
\begin{enumerate}[leftmargin=2em]
\item Failure of weakening creates horns: a judgment derivable in $\Gamma$ may be gapped in $\Gamma, A$
\item Failure of contraction creates horns: a judgment derivable with $A, A$ may be gapped with single $A$
\item Failure of exchange creates horns: a judgment derivable with $A, B$ may be gapped with $B, A$
\end{enumerate}
\end{proposition}

\begin{proof}
Each structural rule, when valid, provides a substitution-like operation. When the rule is restricted, the corresponding transport may fail. The failure is witnessed by the rule violation.
\end{proof}

\subsection{Bunched Implications}

\begin{definition}[BI Logic]\label{def:bi}
\keyterm{Bunched implications} (BI) is a substructural logic with two context-forming operations:
\begin{itemize}[leftmargin=2em]
\item Additive (``,''\!): $\Gamma, \Delta$ shares resources
\item Multiplicative (``$;$''): $\Gamma ; \Delta$ separates resources
\end{itemize}
and corresponding implications $\to$ (additive) and $\mathrel{-\mkern-6mu*}$ (multiplicative, ``magic wand'').
\end{definition}

BI has rich transport structure:

\begin{example}[Separation Failure]\label{ex:separation}
Let $J = $ ``$x$ and $y$ point to disjoint heap regions.''

In context $\Gamma = x : \mathsf{Ptr}, y : \mathsf{Ptr}$, this is underdetermined—$x$ and $y$ might alias.

In context $\Gamma ; \Delta = (x : \mathsf{Ptr}) ; (y : \mathsf{Ptr})$, the separation guarantees disjointness, so $J$ is derivable.

The ``inclusion'' $\Gamma ; \Delta \to \Gamma, \Delta$ is a context morphism, but transport of $J$ along it is gapped: the separation required for $J$ is lost.
\end{example}

\section{Modal Type Theory}

Modal type theories provide another source of transport horns.

\begin{definition}[Modal Context]\label{def:modal-context}
A \keyterm{modal context} is a context with multiple zones, each governed by a modality:
\[
\Gamma = \Gamma_1 \mid \Gamma_2 \mid \cdots \mid \Gamma_n
\]
Variables in different zones have different availability.
\end{definition}

\begin{example}[Necessity Modality]\label{ex:necessity}
In S4-style modal type theory, contexts have the form $\Gamma \mid \Delta$ where:
\begin{itemize}[leftmargin=2em]
\item Variables in $\Gamma$ are ``necessary''—available in all worlds
\item Variables in $\Delta$ are ``contingent''—available only in the current world
\end{itemize}

The $\Box$ modality requires its argument to depend only on $\Gamma$:
\[
\frac{\Gamma \mid \cdot \;\vdash\; t : A}{\Gamma \mid \Delta \;\vdash\; \mathsf{box}(t) : \Box A}
\]

A judgment derivable in $\Gamma \mid \cdot$ may not be derivable in $\Gamma \mid \Delta$ if it must be boxed—the contingent variables in $\Delta$ cannot be used.
\end{example}

\begin{definition}[Modal Transport Horn]\label{def:modal-horn}
A \keyterm{modal transport horn} arises when:
\begin{itemize}[leftmargin=2em]
\item A judgment $J$ is derivable in context $\Gamma \mid \Delta$
\item A modality-respecting substitution $\sigma$ changes the modal structure
\item The transported judgment $J[\sigma]$ is gapped due to modal constraints
\end{itemize}
\end{definition}

\section{Dependent Type Theory with Constraints}

We consider dependent types augmented with constraints.

\begin{definition}[Constrained Context]\label{def:constrained}
A \keyterm{constrained context} is a pair $(\Gamma, \Phi)$ where:
\begin{itemize}[leftmargin=2em]
\item $\Gamma$ is an ordinary context
\item $\Phi$ is a set of constraints (equations, inequalities, predicates) assumed to hold
\end{itemize}
\end{definition}

Derivability in $(\Gamma, \Phi)$ may depend on the constraints. Substitutions may invalidate constraints.

\begin{example}[Refinement Types]\label{ex:refinement}
Let $\Gamma = n : \mathbb{N}$ and $\Phi = \{n > 0\}$.

The judgment $\jdg{(\Gamma, \Phi)}{\mathsf{pred}(n) : \mathbb{N}}$ (predecessor of $n$) is derivable—the constraint ensures $n \geq 1$.

Now let $\sigma = [n \mapsto 0]$. This substitution is well-typed but violates the constraint.

The transported judgment is gapped: in context $(\Gamma[\sigma], \Phi[\sigma])$, the term $\mathsf{pred}(0) : \mathbb{N}$ is undefined or erroneous.
\end{example}

The gap witness here is the constraint violation: $0 > 0$ is false.

\section{Proof Irrelevance and Truncation}

Even in standard HoTT, certain constructions create transport-like phenomena.

\begin{definition}[Propositional Truncation]\label{def:truncation}
The \keyterm{propositional truncation} $\|A\|$ of a type $A$ is the type with:
\begin{itemize}[leftmargin=2em]
\item Constructor: $|a| : \|A\|$ for any $a : A$
\item Path: $|a| = |b|$ for any $a, b : A$
\end{itemize}
Thus $\|A\|$ is a proposition (has at most one element up to path).
\end{definition}

Truncation creates obstructions to ``escaping'':

\begin{proposition}[Truncation Obstruction]\label{prop:truncation}
Let $A$ be a type with $a : A$. We have $|a| : \|A\|$.

To construct a term of type $B$ from $|a|$, we need a function $f : A \to B$ where $B$ is a proposition.

If $B$ is not a proposition, there is no general way to construct an element of $B$ from $|a|$—the truncation has ``forgotten'' which $a$ we had.
\end{proposition}

In OHTT terms: we can model this as a ruptured fibration where the projection $\Sigma_{a:A} B(a) \to \|A\|$ has gapped lifts when $B$ is not propositional.

\begin{example}[Choice]\label{ex:choice}
The axiom of choice (in type-theoretic form) says:
\[
\left(\prod_{x:A} \|B(x)\|\right) \to \left\|\prod_{x:A} B(x)\right\|
\]

This is not provable in HoTT without additional axioms. The obstruction is that we cannot coherently lift the family of truncated witnesses to a single witness.

In OHTT, we can express this as: the section problem for a certain fibration over $A$ is gapped (or at least not coherent).
\end{example}

\section{Decidability and Gaps}

When a type theory has decidable type-checking, underivability can always be certified.

\begin{theorem}[Decidability Implies Gap Witnesses]\label{thm:decidability}
Let $\mathcal{T}$ be a type theory with decidable type-checking. Then for any context $\Gamma$ and judgment $J$:
\begin{enumerate}[leftmargin=2em]
\item Either $\jdg{\Gamma}{J}$ is derivable, or
\item There exists a gap witness (the ``no'' output of the decision procedure)
\end{enumerate}
Thus the ruptured derivability fibration has no open judgments: every judgment is either coherent or gapped.
\end{theorem}

\begin{proof}
Run the decision procedure. If it returns ``yes,'' we have a derivation (coherence witness). If it returns ``no,'' we have a certificate of underivability (gap witness).
\end{proof}

For undecidable type theories (most interesting ones), open judgments exist: we may not know whether a given judgment is derivable.

\section{Summary}

Logical structure is naturally modeled as a ruptured derivability fibration:

\begin{itemize}[leftmargin=2em]
\item \textbf{Resource failure} creates transport horns: derivable judgments may not transport along substitutions that mishandle resources

\item \textbf{Substructural restrictions} create systematic gaps: failures of weakening, contraction, exchange are witnessed obstructions

\item \textbf{Modal constraints} create transport horns: modality-respecting operations may block derivability

\item \textbf{Constraint violations} are gap witnesses: substitutions that violate refinement constraints are witnessed as failing

\item \textbf{Truncation} creates obstructions: escape from propositional truncation is gapped for non-propositional targets

\item \textbf{Decidability} eliminates openness: decidable theories have only coherent or gapped judgments, no open ones
\end{itemize}

The gap witnesses in logical horns are formal certificates: resource counts, constraint violations, modal mismatches, decision procedure outputs. They are \textit{proofs of underivability}, not mere absence of proofs.

What OHTT provides is a uniform framework for these phenomena. Resource-sensitive, substructural, modal, and constrained type theories all exhibit transport failure—and OHTT makes this failure first-class, witnessed, and structured.

\begin{aside}
What this lets us say: ``The substitution lemma is not a universal truth about type theory. It is a property of the coherent fragment. In the full ruptured structure, substitution may fail to preserve derivability—and this failure is as mathematically tractable as success. Underivability is not silence; it is testimony.''
\end{aside}

\part{Contrasts}
\label{part:contrasts}


\chapter{OHTT and HoTT}
\label{ch:ohtt-hott}

\begin{flushright}
\textit{The same letters, rearranged.\\
The same structures, ruptured.}\\[0.5ex]
{\small — The relation}
\end{flushright}

\bigskip

We now make precise the relationship between Open Horn Type Theory and Homotopy Type Theory. The slogan is simple: HoTT is the coherent fragment of OHTT. But the details illuminate what each theory can and cannot express.

\section{The Relationship}

\label{princ:ohtt-hott}
OHTT and HoTT share the same raw syntax of dependent type theory, but differ in what they guarantee:
\begin{enumerate}[leftmargin=2em]
\item \textbf{HoTT} includes total transport as a theorem: given $t : B(x)$ and $p : \Id{A}{x}{y}$, the term $\transport{p}{t} : B(y)$ is always derivable. This corresponds semantically to the Kan filling condition.
\item \textbf{OHTT} does not assume total transport. Transport may succeed (coherent), fail with witness (gapped), or be undetermined (open). This corresponds to dropping the Kan condition.
\end{enumerate}
HoTT is thus recovered from OHTT by adjoining a totality principle ensuring all relevant horns have coherent fillers.

This is the precise sense of the slogan: \textit{OHTT = HoTT $-$ Kan Filling + Gap Witnesses}.

\subsection{The Coherent Fragment}

\begin{definition}[Coherent Fragment]\label{def:coherent-fragment-formal}
The \keyterm{coherent fragment} of OHTT is the subcategory $\mathbf{OHTT}^+$ consisting of:
\begin{itemize}[leftmargin=2em]
\item Contexts $\Gamma$ containing only $\vdash^+$ declarations
\item Judgments of the form $\coh{\Gamma}{J}$
\item Derivations using only rules that produce $\vdash^+$ conclusions from $\vdash^+$ premises
\end{itemize}
\end{definition}

\begin{proposition}[Coherent Fragment and HoTT]\label{prop:coherent-hott}
The coherent fragment $\mathbf{OHTT}^+$ is a \textit{weakening} of $\mathbf{HoTT}$: every $\mathbf{OHTT}^+$ derivation is valid, but $\mathbf{OHTT}^+$ does not include the totality principles (transport, Kan filling) that HoTT provides. Adding those principles recovers HoTT.
\end{proposition}

Thus: \textit{HoTT is exactly the part of OHTT that never mentions gaps, augmented with totality.}

\section{What HoTT Cannot Say}

We enumerate the expressive limitations of HoTT that OHTT overcomes.

\begin{proposition}[No Witnessed Non-Identity]\label{prop:no-witnessed}
In HoTT, there is no judgment form for ``$x$ and $y$ are witnessed as non-identical.'' The closest approximation is:
\[
\neg(\Id{A}{x}{y}) \;\defn\; (\Id{A}{x}{y}) \to \Empty
\]
which says: ``any path from $x$ to $y$ leads to absurdity.''
\end{proposition}

This is weaker than witnessed non-identity:
\begin{itemize}[leftmargin=2em]
\item $\neg(\Id{A}{x}{y})$ is a \textit{function type}—to inhabit it, we must construct a function
\item $\gap{\Gamma}{(p : \Id{A}{x}{y})}$ is a \textit{direct witness}—positive evidence of non-coherence
\end{itemize}

\begin{proposition}[No Transport Failure]\label{prop:no-transport-failure}
In HoTT, given $t : B(x)$ and $p : \Id{A}{x}{y}$, the term $\transport{p}{t} : B(y)$ is \textit{always} well-defined. There is no way to express ``transport fails.''
\end{proposition}

In OHTT, we can have $\gap{\Gamma}{(\transport{p}{t} : B(y))}$—witnessed transport failure.

\begin{proposition}[No Obstruction Data]\label{prop:no-obstruction}
In HoTT, obstructions (monodromy, characteristic classes, etc.) are computed externally and represented as ordinary types or terms. There is no internal representation of ``this horn is unfillable.''
\end{proposition}

In OHTT, gap witnesses \textit{are} obstruction data—first-class terms recording why coherence fails.

\section{The Kan Condition}

The technical heart of the difference is the Kan condition.

\begin{definition}[Kan Condition (Semantic)]\label{def:kan-semantic}
A simplicial set $X$ satisfies the \keyterm{Kan condition} if every horn has a filler:
\[
\forall n \geq 1,\; \forall 0 \leq k \leq n,\; \forall h : \hornsimp{n}{k} \to X,\; \exists \sigma : \simp{n} \to X \text{ extending } h
\]
\end{definition}

\begin{theorem}[HoTT Models are Kan]\label{thm:hott-kan}
In the simplicial model of HoTT, every type is interpreted as a Kan complex. The Kan condition is what makes:
\begin{itemize}[leftmargin=2em]
\item Path inversion work (every path has an inverse)
\item Path composition work (paths compose)
\item Transport work (terms move along paths)
\item Higher coherences work (homotopies exist when needed)
\end{itemize}
\end{theorem}

\begin{theorem}[OHTT Models are Ruptured Kan]\label{thm:ohtt-ruptured}
In the simplicial model of OHTT, types are interpreted as ruptured Kan complexes. The Kan condition is replaced by the trichotomy: every horn is coherently filled, gap-witnessed, or open.
\end{theorem}

The Kan condition is the \textit{totality assumption} that HoTT makes. OHTT drops this assumption.

\section{Conservativity}

\begin{theorem}[OHTT is Conservative over HoTT]\label{thm:conservative}
If $J$ is a judgment expressible in HoTT and $\coh{\Gamma}{J}$ is derivable in OHTT (where $\Gamma$ is a coherent context), then $\jdg{\Gamma}{J}$ is derivable in HoTT.
\end{theorem}

\begin{proof}[Proof sketch]
OHTT adds gap judgments and horn types, but does not change the rules for deriving coherence judgments. A derivation of $\coh{\Gamma}{J}$ in OHTT that uses only coherence rules is exactly a HoTT derivation.
\end{proof}

Conservativity means: OHTT does not prove new coherence facts about HoTT-expressible judgments. What it adds is the ability to express and prove \textit{gap} facts.

\section{The Univalence Axiom}

Univalence is the characteristic axiom of HoTT:
\[
\mathsf{ua} : (A \simeq B) \to (\Id{\UU}{A}{B})
\]

Together with function extensionality, it implies that equivalences \textit{are} identities (in the universe).

\begin{proposition}[Univalence in OHTT]\label{prop:univalence-ohtt}
Univalence can be adopted in OHTT as an axiom about coherence:
\[
\coh{}{\mathsf{ua} : (A \simeq B) \to (\Id{\UU}{A}{B})}
\]
This asserts that every equivalence coherently induces a path.
\end{proposition}

With univalence, OHTT includes HoTT as its coherent fragment. Without univalence, OHTT includes MLTT.

\begin{remark}
Univalence does not eliminate gaps. Even with univalence, we can have:
\begin{itemize}[leftmargin=2em]
\item Gapped paths: $\gap{}{(p : \Id{A}{x}{y})}$ for non-equivalent $x, y$
\item Gapped transport: $\gap{}{(\transport{p}{t} : B(y))}$ when transport is obstructed
\item Gapped equivalences: $\gap{}{(e : A \simeq B)}$ for non-equivalent types
\end{itemize}
Univalence tells us that equivalences give paths; it does not tell us which paths or equivalences exist.
\end{remark}

\section{Higher Inductive Types}

Higher inductive types (HITs) are types specified with path constructors, not just point constructors.

\begin{example}[Circle]\label{ex:circle-hit}
The circle $S^1$ is the HIT with:
\begin{itemize}[leftmargin=2em]
\item Point constructor: $\mathsf{base} : S^1$
\item Path constructor: $\mathsf{loop} : \Id{S^1}{\mathsf{base}}{\mathsf{base}}$
\end{itemize}
\end{example}

\begin{proposition}[HITs in OHTT]\label{prop:hits-ohtt}
Higher inductive types can be defined in OHTT. The path constructors produce coherence witnesses:
\[
\coh{}{\mathsf{loop} : \Id{S^1}{\mathsf{base}}{\mathsf{base}}}
\]
\end{proposition}

OHTT extends this with the possibility of \textit{gap constructors}:

\begin{definition}[Ruptured Higher Inductive Type]\label{def:ruptured-hit}
A \keyterm{ruptured higher inductive type} is a type specified with:
\begin{itemize}[leftmargin=2em]
\item Point constructors (as usual)
\item Path constructors producing coherence witnesses
\item Gap constructors producing gap witnesses
\end{itemize}
\end{definition}

\begin{example}[Ruptured Interval]\label{ex:ruptured-interval}
A ruptured interval $I_\bot$ might have:
\begin{itemize}[leftmargin=2em]
\item Points: $0, 1 : I_\bot$
\item Gap: $\gap{}{(p : \Id{I_\bot}{0}{1})}$
\end{itemize}
This is an ``interval'' where the endpoints are witnessed as non-connected.
\end{example}

Ruptured HITs can model spaces with intrinsic obstructions—not just topological spaces, but spaces-with-certified-gaps.

\section{Models}

We summarize the model theory.

\begin{center}
\small
\begin{tabular}{l|l|l}
\textbf{Theory} & \textbf{Models} & \textbf{Key Property} \\
\hline
MLTT & CwFs & Substitution \\
HoTT & Kan complexes & All horns fill \\
OHTT & Ruptured Kan complexes & Trichotomy \\
\end{tabular}
\end{center}

\begin{theorem}[Model Inclusions]\label{thm:model-inclusions}
\[
\{\text{Kan complexes}\} \subset \{\text{Ruptured Kan complexes}\}
\]
The inclusion sends a Kan complex $X$ to the ruptured Kan complex $(X, X_\bullet, \emptyset)$ where all simplices are coherent and no horns are gapped.
\end{theorem}

The model theory reflects the syntactic relationship: HoTT models are special (maximally coherent) OHTT models.

\section{Summary}

The relationship between OHTT and HoTT:

\begin{itemize}[leftmargin=2em]
\item \textbf{Embedding}: HoTT embeds faithfully in OHTT via $\jdg{\Gamma}{J} \mapsto \coh{\Gamma}{J}$

\item \textbf{Coherent fragment}: HoTT is exactly the $\vdash^+$-only part of OHTT

\item \textbf{Conservativity}: OHTT proves no new coherence facts about HoTT judgments

\item \textbf{Expressiveness}: OHTT can express witnessed non-identity, transport failure, and obstruction data; HoTT cannot

\item \textbf{Kan condition}: HoTT assumes all horns fill; OHTT allows gaps and openness

\item \textbf{Univalence}: Compatible with OHTT; does not eliminate gaps

\item \textbf{HITs}: OHTT extends HITs with gap constructors (ruptured HITs)

\item \textbf{Models}: Kan complexes are maximally-coherent ruptured Kan complexes
\end{itemize}

The slogan: \textit{HoTT is OHTT with Gap $= \emptyset$}. Or equivalently: \textit{OHTT is HoTT without the Kan condition}.

\begin{aside}
The relationship is not antagonistic. HoTT is the right theory for reasoning about spaces where all paths exist and all homotopies cohere. OHTT is the extension for reasoning about spaces with obstructions—where some paths are blocked, some transports fail, and we want to track these failures as positive data. One does not replace the other; they are the coherent and ruptured fragments of a single larger theory.
\end{aside}


\chapter{Negation and Paraconsistency}
\label{ch:negation}

\begin{flushright}
\textit{Gap is not negation.\\
Negation is not gap.}\\[0.5ex]
{\small — The distinction}
\end{flushright}

\bigskip

We now situate OHTT relative to logical traditions that have grappled with negation, contradiction, and non-classical truth values. The gap judgment $\vdash^-$ is none of the standard notions of negation—but it has relatives in the literature that illuminate its nature.

\section{Classical Negation}

\begin{definition}[Classical Negation]\label{def:classical-neg}
In classical logic, negation $\neg P$ is characterized by:
\begin{itemize}[leftmargin=2em]
\item \textbf{Excluded middle}: $P \vee \neg P$ is a theorem
\item \textbf{Double negation elimination}: $\neg\neg P \to P$ is a theorem
\item \textbf{Semantic}: $\neg P$ is true iff $P$ is false
\end{itemize}
\end{definition}

Classical negation is the \textit{complement} of truth: $\neg P$ holds exactly when $P$ fails.

\begin{proposition}[Gap is not Classical Negation]\label{prop:gap-not-classical}
In OHTT:
\begin{enumerate}[leftmargin=2em]
\item There is no excluded middle: a judgment may be neither coherent nor gapped (it may be open)
\item There is no double negation elimination: $\gap{}{\gap{}{J}}$ (a gap-of-gap, if it made sense) does not give $\coh{}{J}$
\item Gap is not the complement of coherence: open judgments are neither
\end{enumerate}
\end{proposition}

The trichotomy (coherent / gapped / open) is fundamentally incompatible with classical two-valued logic.

\section{Intuitionistic Negation}

\begin{definition}[Intuitionistic Negation]\label{def:int-neg}
In intuitionistic logic, negation is defined:
\[
\neg P \;\defn\; P \to \bot
\]
A proof of $\neg P$ is a function that transforms any proof of $P$ into a proof of absurdity.
\end{definition}

Intuitionistic negation is \textit{constructive}: to prove $\neg P$, we must show that $P$ leads to contradiction. But it is also \textit{derivative}: negation is defined in terms of implication and absurdity.

\begin{proposition}[Gap is not Intuitionistic Negation]\label{prop:gap-not-int}
In OHTT:
\begin{enumerate}[leftmargin=2em]
\item Gap is primitive, not defined via implication
\item $\gap{\Gamma}{J}$ does not mean ``$\coh{\Gamma}{J}$ implies absurdity''
\item A gap witness is direct evidence, not a refutation function
\end{enumerate}
\end{proposition}

The difference is subtle but crucial:
\begin{itemize}[leftmargin=2em]
\item Intuitionistic $\neg P$: ``Show me a proof of $P$, and I'll derive a contradiction''
\item OHTT $\gap{}{J}$: ``I have witnessed that $J$ does not cohere'' (no hypothetical, no derivation)
\end{itemize}

\subsection{Strong Negation}

The closest intuitionistic relative is Nelson's \textit{strong negation}.

\begin{definition}[Strong Negation]\label{def:strong-neg}
In Nelson's constructive logic with strong negation, there are two negations:
\begin{itemize}[leftmargin=2em]
\item \textbf{Weak negation}: $\neg P \defn P \to \bot$ (standard intuitionistic)
\item \textbf{Strong negation}: $\sim P$ (primitive, with its own rules)
\end{itemize}
Strong negation satisfies: $\sim\sim P \leftrightarrow P$ and $\sim(P \wedge Q) \leftrightarrow (\sim P \vee \sim Q)$.
\end{definition}

\begin{proposition}[Gap resembles Strong Negation]\label{prop:gap-strong}
OHTT gap has affinities with strong negation:
\begin{enumerate}[leftmargin=2em]
\item Both are primitive (not defined via implication)
\item Both provide ``positive'' refutation
\item Both coexist with a weaker notion
\end{enumerate}
However, gap differs:
\begin{enumerate}[leftmargin=2em]
\item Gap is proof-relevant (different gap witnesses may exist)
\item Gap is at the judgment level, not the proposition level
\item Gap creates horns (compositional tension), not just negated formulas
\end{enumerate}
\end{proposition}

Nelson's logic is perhaps the closest classical precedent to OHTT's gap—but OHTT embeds the idea in a type-theoretic, proof-relevant, higher-categorical framework.

\section{Paraconsistent Logic}

Paraconsistent logics allow contradictions without explosion.

\begin{definition}[Paraconsistency]\label{def:paraconsistent}
A logic is \keyterm{paraconsistent} if the principle of explosion fails:
\[
P, \neg P \nvdash Q \quad \text{(for arbitrary $Q$)}
\]
Contradictions are tolerated; they do not make everything true.
\end{definition}

\begin{proposition}[OHTT is not Paraconsistent]\label{prop:ohtt-not-para}
OHTT is \textbf{not} paraconsistent:
\begin{enumerate}[leftmargin=2em]
\item The Exclusion Law forbids $\coh{\Gamma}{J}$ and $\gap{\Gamma}{J}$ simultaneously
\item There are no contradictions to tolerate—coherence and gap are mutually exclusive by axiom
\end{enumerate}
\end{proposition}

The distinction:
\begin{itemize}[leftmargin=2em]
\item Paraconsistent logics: ``$P$ and $\neg P$ can both hold; we contain the damage''
\item OHTT: ``Coherence and gap cannot both hold; but neither may hold (openness)''
\end{itemize}

Paraconsistency is about tolerating contradiction. OHTT is about the trichotomy \textit{before} contradiction arises.

\section{Many-Valued and Bilattice Logics}

\begin{definition}[Belnap's Four-Valued Logic]\label{def:belnap}
Belnap's logic has four truth values:
\begin{itemize}[leftmargin=2em]
\item \textbf{T}: true (told true, not told false)
\item \textbf{F}: false (told false, not told true)
\item \textbf{B}: both (told both true and false)
\item \textbf{N}: neither (told neither)
\end{itemize}
These form a bilattice with two orderings: truth order and information order.
\end{definition}

\begin{proposition}[OHTT and Belnap]\label{prop:ohtt-belnap}
OHTT's trichotomy resembles Belnap's four values minus ``both'':
\begin{center}
\begin{tabular}{l|l}
\textbf{OHTT} & \textbf{Belnap} \\
\hline
Coherent ($\vdash^+$) & T (true) \\
Gapped ($\vdash^-$) & F (false) \\
Open (neither) & N (neither) \\
— & B (both) \\
\end{tabular}
\end{center}
The Exclusion Law rules out ``both''—OHTT has three states, not four.
\end{proposition}

Belnap's logic is designed for databases that may have contradictory information. OHTT is designed for type theory where contradiction is absurdity—but where \textit{absence of information} (openness) is meaningful.

\subsection{Information Order}

Belnap's bilattice has an \textit{information order}: N $<$ T, F $<$ B. More information moves us up the order.

\begin{proposition}[OHTT Information Order]\label{prop:ohtt-info}
OHTT has a natural information order:
\[
\text{Open} \;<\; \text{Coherent}, \quad \text{Open} \;<\; \text{Gapped}
\]
Open judgments have less information (unwitnessed). Coherent and gapped judgments have more (witnessed). There is no ``both'' because Exclusion forbids it.
\end{proposition}

This partial order captures: witnessing (either coherence or gap) is information gain; openness is the information-minimal state.

\section{Signed and Polarized Logics}

\begin{definition}[Signed Formula]\label{def:signed}
In signed logics, formulas carry explicit signs: $+P$ (assert $P$) and $-P$ (deny $P$). Derivations track signs.
\end{definition}

OHTT's two turnstiles ($\vdash^+$ and $\vdash^-$) resemble signed derivability:
\begin{itemize}[leftmargin=2em]
\item $\coh{\Gamma}{J}$ is like deriving $+J$
\item $\gap{\Gamma}{J}$ is like deriving $-J$
\end{itemize}

\begin{definition}[Polarized Logic]\label{def:polarized}
In Girard's polarized logic (LC, LU), formulas have polarity (positive/negative) governing their decomposition:
\begin{itemize}[leftmargin=2em]
\item Positive formulas decompose eagerly (synchronous)
\item Negative formulas decompose lazily (asynchronous)
\end{itemize}
\end{definition}

OHTT's polarity is different: it marks the \textit{witnessing mode}, not the formula's logical structure. But the idea that judgments carry explicit polarity is shared.

\section{Relevant Logic}

\begin{definition}[Relevant Logic]\label{def:relevant}
Relevant logics require that premises be ``used'' in derivations:
\[
P \to Q \text{ requires that } P \text{ be relevant to } Q
\]
This blocks ``paradoxes of implication'' like $P \to (Q \to Q)$ from arbitrary $P$.
\end{definition}

\begin{proposition}[OHTT and Relevance]\label{prop:ohtt-relevance}
OHTT does not impose relevance constraints on coherence judgments. However:
\begin{enumerate}[leftmargin=2em]
\item Gap witnesses \textit{are} relevant: a gap witness for $J$ must actually witness why $J$ fails
\item Horn inhabitants \textit{are} relevant: the coherences and gap in a horn are structurally connected
\end{enumerate}
\end{proposition}

Relevance in OHTT is semantic (built into what gap witnesses are) rather than syntactic (imposed by logical rules).

\section{The Distinctive Position of OHTT}

We summarize OHTT's position in the landscape:

\begin{center}
\begin{tabular}{l|c|c|c}
\textbf{Feature} & \textbf{Classical} & \textbf{Intuit.} & \textbf{OHTT} \\
\hline
Excluded middle & Yes & No & No \\
Primitive negation & Yes ($\neg$) & No ($\to\bot$) & Yes ($\vdash^-$) \\
Proof-relevant & No & Yes & Yes \\
Third state & No & No* & Yes (open) \\
Explosion & Yes & Yes & N/A (Excl.) \\
\end{tabular}
\end{center}

*Intuitionistic logic lacks excluded middle but does not have a third truth value; ``neither provable nor refutable'' is epistemic, not semantic.

\begin{theorem}[OHTT Distinctiveness]\label{thm:distinctive}
OHTT is characterized by the conjunction:
\begin{enumerate}[leftmargin=2em]
\item Proof-relevant judgments
\item Primitive gap (not defined via implication)
\item Trichotomy (coherent / gapped / open)
\item Exclusion (coherent and gapped are incompatible)
\item Horn structure (compositional tension)
\end{enumerate}
To our knowledge, no standard logic in the literature combines all five features.
\end{theorem}

\section{Summary}

Gap is not negation:

\begin{itemize}[leftmargin=2em]
\item \textbf{Not classical}: no excluded middle, no complement semantics
\item \textbf{Not intuitionistic}: primitive, not defined via $\to\bot$
\item \textbf{Not paraconsistent}: Exclusion forbids ``both,'' not tolerates it
\item \textbf{Resembles strong negation}: primitive refutation, positive evidence
\item \textbf{Resembles Belnap minus ``both''}: three-valued, information-ordered
\end{itemize}

OHTT occupies a distinctive position: constructive (proof-relevant), three-valued (trichotomy), consistent (Exclusion), and compositional (horns). It is not a negation logic but a \textit{witnessing} logic—where both success and failure are first-class data.

\begin{aside}
The point is not that OHTT is better than classical or intuitionistic logic. Each is suited to its domain. Classical logic is for mathematics where excluded middle holds. Intuitionistic logic is for constructive reasoning where proofs are programs. OHTT is for domains with intrinsic obstructions—where we need to say not just what holds, but what is witnessed as blocked.
\end{aside}


\chapter*{Epilogue: What Remains Open}
\addcontentsline{toc}{chapter}{Epilogue: What Remains Open}
\label{ch:epilogue}

\begin{flushright}
\textit{A theory is not a tomb.\\
It is a door.}\\[0.5ex]
{\small — On incompleteness}
\end{flushright}

\bigskip

We have presented Open Horn Type Theory as a self-contained formal system: the calculus (Part~I), the geometric model (Part~II), the worked obstructions (Part~III), and the contrasts with existing logics (Part~IV). But a theory worth its name opens more questions than it closes. We conclude by indicating what remains.

\section*{Computational Content}

OHTT as presented is a \textit{logic}, not a \textit{programming language}. The coherent fragment (HoTT/MLTT) has well-understood computational content: terms compute, types classify programs, the correspondence with $\lambda$-calculus is tight.

What is the computational content of gaps?

One direction: gap witnesses as \textit{certificates}. A gap witness $\omega : \gap{\Gamma}{J}$ could be a data structure that certifies why $J$ fails—a counterexample, a resource-insufficiency proof, an obstruction cocycle. This connects to certified programming, proof-carrying code, and SMT solvers.

Another direction: gap witnesses as \textit{exceptions} or \textit{effects}. Transport failure could raise an exception; the gap witness is the exception payload. This connects to effect systems, monadic programming, and algebraic effects.

A third direction: gap witnesses as \textit{constraints}. In constraint-based type inference, unification may fail with a witness explaining why. OHTT could model such systems with gaps as failed unification witnesses.

None of these is developed here. The computational interpretation of OHTT is open.

\section*{Dynamics and Time}

OHTT as presented is \textit{static}: contexts, judgments, and witnesses are given timelessly. But the motivating examples—texts evolving, meanings drifting, derivability changing with resources—are inherently \textit{dynamic}.

A dynamic OHTT would have:
\begin{itemize}[leftmargin=2em]
\item Time-indexed contexts: $\Gamma_t$ at time $t$
\item Judgment evolution: $\coh{\Gamma_t}{J}$ may become $\gap{\Gamma_{t'}}{J}$ or vice versa
\item Persistent horns: horns that survive across time, tracking stable tensions
\item A scheduler: an agent that decides which judgments to witness, when
\end{itemize}

The scheduler is particularly important. In the static theory, the trichotomy (coherent/gapped/open) is a fact. In the dynamic theory, it is a \textit{choice}: which horns to fill, which to leave open, which to witness as gapped. A healthy scheduler neither compulsively closes all gaps nor refuses all coherence. The dynamics of witnessing—when to press for closure, when to dwell in openness—is not part of the calculus. It is a separate theory, building on OHTT but not reducible to it.

This is open.

\section*{Empirical Instantiation}

The semantic examples of Chapter~\ref{ch:semantic} gesture at a connection to empirical linguistics and natural language processing. But the connection is not made precise.

An empirical instantiation would:
\begin{itemize}[leftmargin=2em]
\item Define $L$ (lexical space) and $M$ (meaning space) concretely—perhaps as embedding spaces from neural language models
\item Define coherence and gap in terms of geometric structure—perhaps via persistent homology, Čech complexes, or Vietoris-Rips filtrations
\item Extract horns from data—identify where lexical paths fail to lift to meaning paths
\item Validate against linguistic judgments—do the computed gaps correspond to human polysemy intuitions?
\end{itemize}

This would connect OHTT to topological data analysis, distributional semantics, and computational linguistics. The tools exist (persistent homology, word embeddings, attention patterns). The bridge is not built.

This is open.

\section*{Higher Gaps}

OHTT has gaps at all levels: gaps between terms, gaps between paths, gaps between homotopies. The higher structure is sketched but not fully developed.

Questions:
\begin{itemize}[leftmargin=2em]
\item What is the homotopy theory of ruptured Kan complexes? Is there a model structure?
\item What are the higher obstruction invariants? Characteristic classes for ruptured fibrations?
\item Is there a ruptured $\infty$-topos? What are its internal logic and type theory?
\item How does truncation interact with gaps? Is there a ruptured $n$-type theory?
\end{itemize}

The foundations are laid. The higher theory is open.

\section*{Categorical Semantics}

HoTT has multiple categorical semantics: $(\infty,1)$-toposes, model categories, type-theoretic fibration categories. OHTT should have analogous semantics, but they are not developed here.

Candidates:
\begin{itemize}[leftmargin=2em]
\item \textbf{Ruptured $(\infty,1)$-categories}: $\infty$-categories with marked non-invertible morphisms (gaps)
\item \textbf{Stratified fibration categories}: categories with families augmented by gap structure
\item \textbf{Bilattice-enriched categories}: categories enriched over a bilattice of truth/information values
\end{itemize}

The right categorical semantics would clarify OHTT's relationship to existing categorical logic.

This is open.

\section*{Implementations}

HoTT has proof assistants: Agda with cubical mode, Coq with HoTT library, Arend, cubicaltt. OHTT has none.

An OHTT proof assistant would need:
\begin{itemize}[leftmargin=2em]
\item Syntax for $\vdash^+$ and $\vdash^-$ judgments
\item Type-checking that respects Exclusion
\item Horn types as first-class constructions
\item Perhaps: tactics for constructing gap witnesses
\end{itemize}

The implementation is open.

\section*{Beyond}

Finally, there are directions that extend beyond the formal:

\begin{itemize}[leftmargin=2em]
\item \textbf{The ethics of witnessing}: When should we witness a gap rather than leave it open? When should we seek coherence rather than rest in rupture? These are not formal questions, but they arise from the formalism.

\item \textbf{The theology of the open}: If coherence is presence and gap is witnessed absence, what is openness? The space of the not-yet-witnessed has resonances with apophatic theology, with the Kabbalistic \textit{tzimtzum}, with the Sufi station of \textit{fanā'}. These are not formalizable, but the formalism makes room for them.

\item \textbf{The self as homotopy colimit}: If a self is the homotopy colimit of its trajectory through time, then different paths to the same ``endpoint'' yield different selves. The open horn—the path not taken, the coherence not forced—is part of the self's constitution. This is speculative, but it is where the formalism points.
\end{itemize}

We leave these open, not because they are unimportant, but because they are beyond the scope of a logic monograph. OHTT is a foundation. What is built on it is another matter.

\bigskip

\begin{center}
*\quad*\quad*
\end{center}

\bigskip

The open horn is not failure. It is the shape of what does not close.

The theory is complete. The work begins.

\backmatter

\nocite{*}
\bibliographystyle{alpha}
\bibliography{references}

@book{hottbook,
  author = {{The Univalent Foundations Program}},
  title = {Homotopy Type Theory: Univalent Foundations of Mathematics},
  publisher = {Institute for Advanced Study},
  year = {2013},
  url = {https://homotopytypetheory.org/book}
}

@article{voevodsky2010,
  author = {Voevodsky, Vladimir},
  title = {Univalent Foundations Project},
  journal = {NSF Grant Application},
  year = {2010},
  note = {Available at \url{https://www.math.ias.edu/vladimir/Site3/Univalent_Foundations.html}}
}

@inproceedings{awodey-warren2009,
  author = {Awodey, Steve and Warren, Michael A.},
  title = {Homotopy theoretic models of identity types},
  booktitle = {Mathematical Proceedings of the Cambridge Philosophical Society},
  volume = {146},
  number = {1},
  pages = {45--55},
  year = {2009},
  publisher = {Cambridge University Press}
}

@article{lumsdaine2010,
  author = {Lumsdaine, Peter LeFanu},
  title = {Weak $\omega$-categories from intensional type theory},
  journal = {Logical Methods in Computer Science},
  volume = {6},
  number = {3},
  year = {2010}
}

@inproceedings{cchm2018,
  author = {Cohen, Cyril and Coquand, Thierry and Huber, Simon and M{\"o}rtberg, Anders},
  title = {Cubical Type Theory: A Constructive Interpretation of the Univalence Axiom},
  booktitle = {21st International Conference on Types for Proofs and Programs (TYPES 2015)},
  year = {2018},
  pages = {5:1--5:34},
  publisher = {Schloss Dagstuhl}
}

@book{martin-lof1984,
  author = {Martin-L{\"o}f, Per},
  title = {Intuitionistic Type Theory},
  publisher = {Bibliopolis},
  address = {Naples},
  year = {1984},
  note = {Notes by Giovanni Sambin}
}

@incollection{martin-lof1975,
  author = {Martin-L{\"o}f, Per},
  title = {An Intuitionistic Theory of Types: Predicative Part},
  booktitle = {Logic Colloquium '73},
  editor = {Rose, H. E. and Shepherdson, J. C.},
  publisher = {North-Holland},
  year = {1975},
  pages = {73--118}
}

@book{nordstrom1990,
  author = {Nordstr{\"o}m, Bengt and Petersson, Kent and Smith, Jan M.},
  title = {Programming in Martin-L{\"o}f's Type Theory},
  publisher = {Oxford University Press},
  year = {1990}
}

@book{goerss-jardine,
  author = {Goerss, Paul G. and Jardine, John F.},
  title = {Simplicial Homotopy Theory},
  publisher = {Birkh{\"a}user},
  year = {1999},
  series = {Progress in Mathematics},
  volume = {174}
}

@book{may1967,
  author = {May, J. Peter},
  title = {Simplicial Objects in Algebraic Topology},
  publisher = {University of Chicago Press},
  year = {1967}
}

@book{friedman2020,
  author = {Friedman, Greg},
  title = {Survey Article: An Elementary Illustrated Introduction to Simplicial Sets},
  journal = {Rocky Mountain Journal of Mathematics},
  volume = {42},
  number = {2},
  pages = {353--423},
  year = {2012}
}

@book{lurie2009,
  author = {Lurie, Jacob},
  title = {Higher Topos Theory},
  publisher = {Princeton University Press},
  year = {2009},
  series = {Annals of Mathematics Studies},
  volume = {170}
}

@book{steenrod1951,
  author = {Steenrod, Norman},
  title = {The Topology of Fibre Bundles},
  publisher = {Princeton University Press},
  year = {1951}
}

@book{husemoller1994,
  author = {Husem{\"o}ller, Dale},
  title = {Fibre Bundles},
  edition = {3rd},
  publisher = {Springer},
  year = {1994},
  series = {Graduate Texts in Mathematics},
  volume = {20}
}

@book{hatcher2002,
  author = {Hatcher, Allen},
  title = {Algebraic Topology},
  publisher = {Cambridge University Press},
  year = {2002}
}

@article{whitehead1949,
  author = {Whitehead, J. H. C.},
  title = {Combinatorial homotopy. {I}},
  journal = {Bulletin of the American Mathematical Society},
  volume = {55},
  number = {3},
  pages = {213--245},
  year = {1949}
}

@article{nelson1949,
  author = {Nelson, David},
  title = {Constructible falsity},
  journal = {Journal of Symbolic Logic},
  volume = {14},
  number = {1},
  pages = {16--26},
  year = {1949}
}

@book{troelstra-schwichtenberg,
  author = {Troelstra, Anne S. and Schwichtenberg, Helmut},
  title = {Basic Proof Theory},
  edition = {2nd},
  publisher = {Cambridge University Press},
  year = {2000}
}

@article{almukdad-nelson1984,
  author = {Almukdad, Ahmad and Nelson, David},
  title = {Constructible falsity and inexact predicates},
  journal = {Journal of Symbolic Logic},
  volume = {49},
  number = {1},
  pages = {231--233},
  year = {1984}
}

@article{belnap1977,
  author = {Belnap, Nuel D.},
  title = {A useful four-valued logic},
  journal = {Modern Uses of Multiple-Valued Logic},
  pages = {5--37},
  year = {1977},
  publisher = {Reidel}
}

@book{priest2006,
  author = {Priest, Graham},
  title = {In Contradiction: A Study of the Transconsistent},
  edition = {2nd},
  publisher = {Oxford University Press},
  year = {2006}
}

@article{fitting1991,
  author = {Fitting, Melvin},
  title = {Bilattices and the semantics of logic programming},
  journal = {Journal of Logic Programming},
  volume = {11},
  number = {2},
  pages = {91--116},
  year = {1991}
}

@incollection{ginsberg1988,
  author = {Ginsberg, Matthew L.},
  title = {Multivalued logics: A uniform approach to reasoning in artificial intelligence},
  booktitle = {Computational Intelligence},
  volume = {4},
  number = {3},
  pages = {265--316},
  year = {1988}
}

@article{girard1987,
  author = {Girard, Jean-Yves},
  title = {Linear logic},
  journal = {Theoretical Computer Science},
  volume = {50},
  number = {1},
  pages = {1--101},
  year = {1987}
}

@book{restall2000,
  author = {Restall, Greg},
  title = {An Introduction to Substructural Logics},
  publisher = {Routledge},
  year = {2000}
}

@article{ohearn-pym1999,
  author = {O'Hearn, Peter W. and Pym, David J.},
  title = {The logic of bunched implications},
  journal = {Bulletin of Symbolic Logic},
  volume = {5},
  number = {2},
  pages = {215--244},
  year = {1999}
}

@inproceedings{reynolds2002,
  author = {Reynolds, John C.},
  title = {Separation logic: A logic for shared mutable data structures},
  booktitle = {Proceedings of the 17th Annual IEEE Symposium on Logic in Computer Science},
  pages = {55--74},
  year = {2002}
}

@inproceedings{pfenning-davies2001,
  author = {Pfenning, Frank and Davies, Rowan},
  title = {A judgmental reconstruction of modal logic},
  booktitle = {Mathematical Structures in Computer Science},
  volume = {11},
  number = {4},
  pages = {511--540},
  year = {2001}
}

@phdthesis{licata2011,
  author = {Licata, Daniel R.},
  title = {Dependently Typed Programming with Domain-Specific Logics},
  school = {Carnegie Mellon University},
  year = {2011}
}

@article{shulman2018,
  author = {Shulman, Michael},
  title = {Brouwer's fixed-point theorem in real-cohesive homotopy type theory},
  journal = {Mathematical Structures in Computer Science},
  volume = {28},
  number = {6},
  pages = {856--941},
  year = {2018}
}

@book{pustejovsky1995,
  author = {Pustejovsky, James},
  title = {The Generative Lexicon},
  publisher = {MIT Press},
  year = {1995}
}

@article{kilgarriff1992,
  author = {Kilgarriff, Adam},
  title = {Polysemy},
  journal = {Language and Cognitive Processes},
  year = {1992},
  note = {PhD Thesis, University of Sussex}
}

@book{cruse1986,
  author = {Cruse, D. A.},
  title = {Lexical Semantics},
  publisher = {Cambridge University Press},
  year = {1986}
}

@article{apresjan1974,
  author = {Apresjan, Juri D.},
  title = {Regular polysemy},
  journal = {Linguistics},
  volume = {12},
  number = {142},
  pages = {5--32},
  year = {1974}
}

@article{carlsson2009,
  author = {Carlsson, Gunnar},
  title = {Topology and data},
  journal = {Bulletin of the American Mathematical Society},
  volume = {46},
  number = {2},
  pages = {255--308},
  year = {2009}
}

@article{edelsbrunner-harer2008,
  author = {Edelsbrunner, Herbert and Harer, John},
  title = {Persistent homology---a survey},
  journal = {Contemporary Mathematics},
  volume = {453},
  pages = {257--282},
  year = {2008}
}

@book{edelsbrunner-harer2010,
  author = {Edelsbrunner, Herbert and Harer, John L.},
  title = {Computational Topology: An Introduction},
  publisher = {American Mathematical Society},
  year = {2010}
}

@article{zomorodian-carlsson2005,
  author = {Zomorodian, Afra and Carlsson, Gunnar},
  title = {Computing persistent homology},
  journal = {Discrete \& Computational Geometry},
  volume = {33},
  number = {2},
  pages = {249--274},
  year = {2005}
}

@book{jacobs1999,
  author = {Jacobs, Bart},
  title = {Categorical Logic and Type Theory},
  publisher = {Elsevier},
  year = {1999},
  series = {Studies in Logic and the Foundations of Mathematics},
  volume = {141}
}

@article{awodey2018,
  author = {Awodey, Steve},
  title = {Natural models of homotopy type theory},
  journal = {Mathematical Structures in Computer Science},
  volume = {28},
  number = {2},
  pages = {241--286},
  year = {2018}
}

@book{lambek-scott1986,
  author = {Lambek, Joachim and Scott, Philip J.},
  title = {Introduction to Higher Order Categorical Logic},
  publisher = {Cambridge University Press},
  year = {1986}
}

@article{hofmann-streicher1998,
  author = {Hofmann, Martin and Streicher, Thomas},
  title = {The groupoid interpretation of type theory},
  journal = {Twenty-Five Years of Constructive Type Theory},
  pages = {83--111},
  year = {1998},
  publisher = {Oxford University Press}
}

@inproceedings{norell2009,
  author = {Norell, Ulf},
  title = {Dependently typed programming in {Agda}},
  booktitle = {Proceedings of the 4th International Workshop on Types in Language Design and Implementation},
  pages = {1--2},
  year = {2009}
}

@manual{coq2021,
  author = {{The Coq Development Team}},
  title = {The {Coq} Proof Assistant Reference Manual},
  year = {2021},
  url = {https://coq.inria.fr/refman/}
}

@article{vezzosi2021,
  author = {Vezzosi, Andrea and M{\"o}rtberg, Anders and Abel, Andreas},
  title = {Cubical {Agda}: A dependently typed programming language with univalence and higher inductive types},
  journal = {Journal of Functional Programming},
  volume = {31},
  year = {2021}
}


\end{document}